\documentclass{article} 
\usepackage[margin=0.9in]{geometry}
\usepackage[parfill]{parskip}
\usepackage{graphicx}
\usepackage{algorithm,algorithmic}
\usepackage{amsmath, amssymb, amsthm, bm,bbm}
\usepackage{url}
\usepackage{boxedminipage}
\usepackage{wrapfig}
\usepackage{ifthen}
\usepackage{xcolor}
\usepackage{psfrag}
\usepackage{framed}
\usepackage{algorithmic,algorithm}
\usepackage{enumitem}
\usepackage{epstopdf}
\newtheorem{theorem}{Theorem}

\newtheorem{lemma}{Lemma}

        {\hspace*{\fill}$\Box$\par}

\newcommand{\ignore}[1]{}

\newcommand{\R}{\mathbb R}
\newcommand{\N}{\mathbb N}

\newcommand{\abs}[1]{\left\lvert #1 \right\rvert}

\newcommand{\infnorm}[1]{\left\lVert #1 \right\rVert_{\infty}}
\newcommand{\iprod}[2]{\left\langle #1,#2 \right\rangle}
\newcommand{\twonorm}[1]{\left\lVert #1 \right\rVert_{2}}
\newcommand{\frob}[1]{\left\lVert #1 \right\rVert_{F}}

\colorlet{shadecolor}{blue!20}

\newcommand{\mB}{{B}}
\newcommand{\A}{{A}}
\newcommand{\defas}{:=}
\newcommand{\wL}{\widetilde{\L}}
\newcommand{\wS}{\widetilde{\S}}
\newcommand{\wM}{\widetilde{\M}}
\newcommand{\trans}[1]{{#1}^{\top}}
\newcommand{\Etn}{\E^{(t+1)}}

\def\tha{{\mbox{\tiny th}}}

\def\viz{{viz.,\ \/}}

\title{Non-convex Robust PCA}
\author{Praneeth Netrapalli \thanks{Microsoft Research, Cambridge MA. Email: praneeth@microsoft.com. Part of the work done while interning at Microsoft Research, India.} \and U N Niranjan \thanks{The University of California at Irvine, CA. Email: un.niranjan@uci.edu. Part of the work done while interning at Microsoft Research, India.} \and Sujay Sanghavi  \thanks{The University of Texas at Austin, TX. Email: sanghavi@mail.utexas.edu} \and Animashree Anandkumar \thanks{The University of California at Irvine, CA. Email: a.anandkumar@uci.edu} \and Prateek Jain \thanks{Microsoft Research, Bangalore, India. Email: prajain@microsoft.com}\\
}

\renewcommand{\L}{L}
\renewcommand{\S}{S}
\renewcommand{\N}{N}
\newcommand{\Lo}{\L^*}
\newcommand{\So}{\S^*}
\newcommand{\No}{\N^*}
\newcommand{\Lt}[1][t]{\L^{(#1)}}
\newcommand{\St}[1][t]{\S^{(#1)}}
\newcommand{\Ltn}{\L^{(t+1)}}
\newcommand{\Stn}{\S^{(t+1)}}
\newcommand{\M}{M}
\newcommand{\Mt}[1][t]{\M^{(#1)}}

\newcommand{\e}{\bm{e}}
\newcommand{\E}{{E}}
\newcommand{\Et}[1][t]{\E^{(#1)}}

\newcommand{\Uo}{U^*}
\newcommand{\uo}{\bm{u}^*}

\newcommand{\Vo}{V^*}
\newcommand{\Sigmao}{\Sigma^*}
\newcommand{\Vot}{(\Vo)^\top}
\newcommand{\ind}[1]{\mathbbm{1}_{\left\{#1\right\}}}
\newcommand{\order}[1]{O\left({#1}\right)}
\newcommand{\ncralgo}{AltProj}

\newcommand{\supp}[1]{\textrm{Supp}\left(#1\right)}
\newcommand{\inv}[1]{{#1}^{-1}}

\begin{document}

\maketitle

\begin{abstract}
We propose a new method for robust PCA -- the task of recovering a low-rank matrix from sparse corruptions that are of unknown value and support. Our method involves alternating between projecting appropriate residuals onto the set of low-rank matrices, and the set of sparse matrices; each projection is {\em non-convex} but easy to compute. In spite of this non-convexity, we establish exact recovery of the low-rank matrix, under the same conditions that are required by existing methods (which are based on convex optimization). For an $m \times n$ input matrix ($m \leq n)$, our method has a running time of $\order{r^2mn}$ per iteration, and needs $\order{\log(1/\epsilon)}$ iterations to reach an accuracy of $\epsilon$. This is close to the running times of simple PCA via the power method, which requires $\order{rmn}$ per iteration, and $\order{\log(1/\epsilon)}$ iterations. In contrast, the existing methods for robust PCA, which are based on convex optimization, have $\order{m^2n}$ complexity per iteration, and take $\order{1/\epsilon}$ iterations, i.e., exponentially more iterations for the same accuracy. 

Experiments on both synthetic and real data establishes the improved speed and accuracy of our method over existing convex implementations.

\end{abstract}

\paragraph{Keywords: }Robust PCA,  matrix decomposition, non-convex methods, alternating projections.

\section{Introduction}

Principal component analysis (PCA) is  a common procedure for preprocessing and denoising, where a low rank approximation to the input matrix (such as the covariance matrix) is carried out. Although PCA is simple to implement via eigen-decomposition, it is sensitive to the presence of outliers, since it attempts to ``force fit'' the outliers to the low rank approximation. To overcome this,  the notion of robust PCA is employed, where the goal is to remove sparse corruptions from an input matrix and obtain a low rank approximation.
Robust PCA has been employed in   a wide range of applications, including background modeling~\cite{li2004statistical},  3d reconstruction~\cite{MobahiZYM11}, robust topic modeling~\cite{shi2013sparse}, and community detection~\cite{chen2012clustering}, and so on.

Concretely, robust PCA refers to the following problem: given an input matrix $M = L^* + S^*$, the goal is to decompose it into sparse $S^*$ and low rank $L^*$ matrices. The seminal works of~\cite{ChandrasekaranSPW11,CandesLMW11} showed that this problem can be provably solved via   convex relaxation methods, under some natural conditions on the low rank and sparse components. While the theory is elegant, in practice, convex techniques are expensive to run on a large scale and have poor convergence rates.
Concretely, for decomposing an $m\times n$ matrix,  say with $m \leq n$, the best specialized implementations
(typically first-order  methods)  have a {\em per-iteration complexity} of $\order{m^2n}$, and require $O(1/\epsilon)$
number of iterations to achieve an error of $\epsilon$.
In contrast, the usual PCA, which carries out a rank-$r$ approximation of the input matrix, has $O(rmn)$ complexity
per iteration -- drastically smaller when $r$ is much smaller than $m,n$. Moreover, PCA requires exponentially fewer
iterations for convergence:  an $\epsilon$ accuracy is achieved with only $\order{\log(1/\epsilon)}$ iterations (assuming constant gap in singular values).  

In this paper, we  design a non-convex algorithm which is  ``best of both the worlds'' and bridges
the gap between (the usual) PCA and convex methods for robust PCA. 
Our method has    low computational complexity similar to PCA (i.e. scaling costs and convergence rates), and at the same time,  
has  provable global convergence guarantees, similar to the convex methods.
Proving global convergence for non-convex methods 
is an exciting recent development in machine learning.
Non-convex alternating minimization techniques have recently shown success in many settings such as matrix completion~\cite{Keshavan2012,jain2013low,hardt2013provable}, phase retrieval \cite{Netrapalli0S13},  dictionary learning~\cite{AgarwalEtal:SparseCoding2013}, tensor decompositions for unsupervised learning~\cite{AnandkumarEtal:tensor12}, and so on. Our current work on the analysis of non-convex methods for robust PCA is an important addition to this growing list.

\subsection{Summary of Contributions}

We propose a simple intuitive algorithm for robust PCA with low per-iteration cost and a fast convergence rate. We prove tight guarantees for recovery of   sparse and low rank components, which match those for the convex methods. 
In the process, we derive novel matrix perturbation bounds, when subject to 
sparse perturbations.
Our experiments reveal significant gains in terms of speed-ups over the
convex relaxation techniques, especially as we scale the size of the input matrices.

Our method consists of simple  alternating (non-convex) projections onto low-rank and sparse matrices.
For an $m \times n $ matrix, our method has a running time of
$O(r^2 m n \log(1/\epsilon))$, where $r$ is the rank of the low rank component.
Thus, our method has a linear convergence rate, i.e. it requires $O(\log (1/\epsilon))$ iterations to achieve an error of $\epsilon$, where $r$ is the rank of the low rank component $L^*$. When the rank $r$ is small, this nearly matches the complexity of PCA, (which is $O(rmn \log (1/\epsilon))$).

We prove recovery of the sparse and low rank components under a set of requirements which are tight  and match those for the convex techniques
(up to constant factors). In particular, under the deterministic sparsity model, where each row and each column of
the sparse matrix $S^*$ has at most $\alpha $ fraction of non-zeros,
we require that $\alpha = \order{1/ (\mu^2 r)}$, where $\mu$ is the incoherence factor (see Section~\ref{sec:analysis}).

In addition to strong theoretical guarantees, in practice, our method enjoys significant advantages over the state-of-art solver for \eqref{eq:conv}, \viz the inexact augmented Lagrange multiplier (IALM) method~\cite{CandesLMW11}. Our method outperforms IALM in all instances, as we vary the sparsity levels, incoherence, and rank, in terms of running time to achieve a fixed level of accuracy. In addition, on a real dataset involving the standard task of foreground-background separation \cite{CandesLMW11}, our method is significantly faster and provides  visually better separation.

\paragraph{Overview of our techniques: }Our proof technique involves 
establishing error contraction with each projection   onto 
the sets of low rank and sparse matrices. We first describe the proof ideas
when  $L^*$ is rank one. The first projection step 
is a hard thresholding procedure on  the input matrix $M$ to remove
large entries and then we perform  rank-$1$ projection of the residual to obtain $L^{(1)}$. Standard
matrix perturbation results (such as Davis-Kahan) provide $\ell_2$ error bounds 
between the singular vectors of $L^{(1)}$  and  $L^*$. 
However, these bounds do not suffice 
for establishing the correctness of our method. 
Since the next step in our method involves
hard thresholding of the residual $M-L^{(1)}$, 
we require element-wise error bounds on our low rank estimate. 
Inspired by the approach of Erd{\H{o}}s et al. \cite{ErdosKYY13},
where they obtain similar element-wise bounds for the eigenvectors of sparse Erd{\H{o}}s--R{\'e}nyi graphs,
we derive these bounds by exploiting the fixed 
point characterization of the eigenvectors\footnote{If the input matrix $M$ 
is not symmetric, we embed it in a symmetric 
matrix and consider the eigenvectors of the corresponding matrix.}. A 
Taylor's series expansion reveals that the perturbation between  the estimated and the true eigenvectors consists of 
bounding the walks in a graph whose adjacency matrix corresponds to (a 
subgraph of) the 
sparse component $S^*$. We then  show that if the graph is sparse enough, then this perturbation 
can be controlled, and thus, the next thresholding step    results in further
error contraction. We  use an induction argument to show that the sparse 
estimate is always contained in the true support of $S^*$, and that there is an error 
contraction in each step. For the case, where $L^*$ has rank $r>1$, our algorithm
proceeds in several stages, where we progressively compute higher rank projections 
which alternate with the hard thresholding steps. In stage $k=[1,2,\ldots, r]$, we 
compute rank-$k$ projections, and show that  
after a sufficient number of alternating projections, we reduce the error to the 
level of $(k+1)^{\tha}$ singular value of $L^*$, using similar arguments as in the rank-$1$ case. We then proceed to performing 
rank-$(k+1)$ projections which alternate with hard thresholding. This 
stage-wise procedure is needed for ill-conditioned matrices, since we 
cannot hope to recover lower eigenvectors in the beginning when there are 
large perturbations. Thus, we establish global 
convergence guarantees for our proposed non-convex robust PCA method.

\subsection{Related Work}
Guaranteed methods for robust PCA have received a lot of attention in the past few years, starting from the seminal works of~\cite{ChandrasekaranSPW11,CandesLMW11}, where they showed   recovery of an incoherent low rank matrix $L^*$ through the following convex relaxation method:
\begin{equation}
  \label{eq:conv}
  \text{Conv-RPCA}:\qquad \min_{L, S} \|\L\|_*+\lambda \|\S\|_1,\qquad \text{s.t.},\ \ \  \M=\L+\S,
\end{equation}
where $\|L\|_*$ denotes the nuclear norm of $L$ (nuclear norm is the sum of singular values). A typical solver for this convex program involves projection on to $\ell_1$ and nuclear norm balls (which are convex sets). Note that the convex method can be viewed as ``soft'' thresholding in the standard and spectral domains, while our method involves hard thresholding in these domains.

\cite{ChandrasekaranSPW11} and~\cite{CandesLMW11} consider two different models of sparsity for $S^*$.
Chandrasekaran et al.~\cite{ChandrasekaranSPW11} consider a deterministic sparsity model, where each row and column
of the $m \times n $ matrix, $S$, has at most $\alpha$ fraction of non-zero entries. For guaranteed recovery, they require $\alpha=\order{1/(\mu^2 r \sqrt{n})}$,
where $\mu$ is the incoherence level of $L^*$, and $r$ is its rank.
Hsu et al.~\cite{hsu2011robust} improve upon this result to obtain guarantees for an optimal sparsity level of
$\alpha=\order{1/(\mu^2 r)}$. This  {\em matches}  the requirements
of our non-convex method for exact recovery. Note that when the rank $r = O(1)$, this allows for a constant fraction of corrupted entries. Cand{\`e}s et al.~\cite{CandesLMW11} consider a different model with random sparsity and additional incoherence constraints, viz., they require $ \|U V^\top\|_\infty< \mu \sqrt{r}/n$. Note that  our assumption of  incoherence, viz., $\|U^{(i)}\|< \mu \sqrt{r/n}$, only yields $\|U V^\top\|_\infty<\mu^2 r/n$. The additional assumption enables \cite{CandesLMW11} to prove exact recovery with a constant fraction of corrupted entries, even when $\Lo$ is nearly full-rank.  We note that removing the $\|U V^\top\|_\infty$ condition for robust PCA would imply solving the planted clique problem when the clique size is less than $\sqrt{n}$~\cite{2013arXiv1310.0154C}. Thus,  our recovery guarantees  are {\em tight} upto constants  without these  additional assumptions.

A number of works have considered modified models under the robust PCA framework, e.g.~\cite{agarwal2012noisy,XuCS12}.
For instance, Agarwal et al.~\cite{agarwal2012noisy} relax the incoherence assumption to a weaker ``diffusivity'' assumption, which bounds the magnitude of the entries in the low rank part, but incurs an additional approximation error.
Xu et al.\cite{XuCS12} impose special sparsity structure where a column can either be non-zero or fully zero.

In terms of state-of-art specialized solvers,~\cite{CandesLMW11} implements the in-exact augmented Lagrangian multipliers (IALM) method and provides guidelines for parameter tuning. Other related methods such as multi-block  alternating directions method of multipliers (ADMM)   have also been considered for robust PCA, e.g.~\cite{wang2013solving}. Recently, a multi-step multi-block stochastic ADMM method was analyzed  for this problem~\cite{SedeghiEtal:ADMM14}, and this requires $1/\epsilon$ iterations to achieve an error of $\epsilon$. In addition, the convergence rate is tight in terms of scaling with respect to problem size $(m,n)$ and sparsity and rank parameters, under random noise models.

There is only one other work which considers a non-convex method for robust PCA~\cite{kyrillidis2012matrix}.
However, their result holds only for significantly more restrictive settings and does not cover the deterministic
sparsity assumption that we study. Moreover, the  projection step in their method can have an arbitrarily large
rank, so the running time is still $O(m^2n)$, which is the same as the convex methods. In contrast, we have an improved running time of $O(r^2 mn)$.




\begin{figure}[th!]
\centering
\includegraphics[width=2.5in]{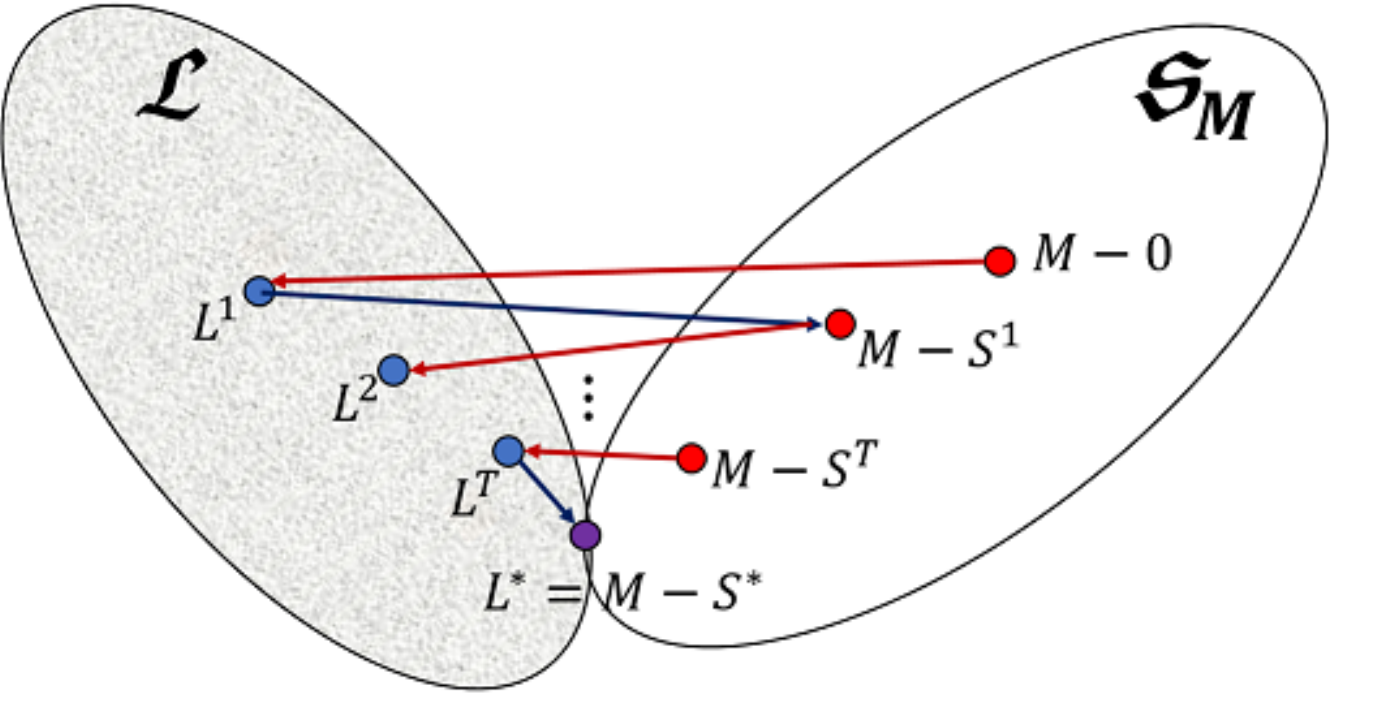}
\caption{Illustration of alternating projections. The goal is to find a matrix $L^*$ which lies in the intersection of two sets: $\mathcal{L}=\{$ set of rank-$r$ matrices$\}$ and $\mathcal{S}_M=\{M-S,\ $ where $S$ is a sparse matrix$\}$. Intuitively, our algorithm alternately projects onto the above two non-convex sets, while appropriately relaxing the rank and the sparsity levels.}\label{fig:path}\vspace*{-15pt}
\end{figure}
\section{Algorithm}\label{sec:prob}




In this section, we present our algorithm for the robust PCA problem.
The robust PCA problem can be formulated as the following optimization problem: find $L, S$ s.t. $\|M-L-S\|_F\leq \epsilon$\footnote{$\epsilon$ is the desired reconstruction error} and 
\begin{enumerate}
\item $L$ lies in the set of low-rank matrices,
\item $S$ lies in the set of sparse matrices.
\end{enumerate}

A natural algorithm for the above problem is to iteratively project $M-L$ onto the set of sparse matrices to update $S$, and then to project $M-S$ onto the set of low-rank matrices to update $L$. Alternatively, one can view the problem as that of finding a matrix $L$ in the intersection of the following two sets: a) $\mathcal{L}=\{$ set of rank-$r$ matrices$\}$, b) $\mathcal{S}_M=\{M-S,\ $ where $S$ is a sparse matrix$\}$. Note that these projections can be done efficiently, even though the sets   are  non-convex. Hard thresholding (HT) is employed  for projections on to sparse matrices,  and singular value decomposition (SVD) is used   for projections on to low rank matrices.\\[0pt]
\textbf{Rank-$1$ case: }We first describe our algorithm for the special case when $L^*$ is rank 1.  Our algorithm  performs an initial hard thresholding  to remove very large entries from input $M$. Note that if we performed the   projection on to rank-$1$ matrices without the initial hard thresholding, we would not make any progress since it is subject to large perturbations. We alternate between computing the rank-$1$ projection of $M-S$, and performing hard thresholding on $M-L$ to remove    entries exceeding a certain threshold. This threshold is gradually decreased as the iterations proceed, and the algorithm is run for a certain number of iterations (which depends on   the desired reconstruction error).\\[0pt] 
\textbf{General rank case: }When $L^*$  has rank $r>1$, a naive extension of our algorithm  consists of alternating projections on to rank-$r$ matrices and sparse matrices. However, such a method has    poor performance on ill-conditioned matrices. This is because  after the initial thresholding of the input matrix $M$, the sparse corruptions in the residual are of the order of the top singular value (with the choice of threshold as specified in the algorithm). When the lower singular values are much smaller, the corresponding singular vectors are subject to relatively large perturbations and thus, we cannot make progress in improving the reconstruction error. To alleviate the dependence on the condition number, we propose an algorithm that proceeds in stages. In the $k^{\tha}$ stage, the algorithm alternates between rank-$k$ projections and hard thresholding for a certain number of iterations. We run the algorithm for $r$ stages, where $r$ is the rank of $L^*$. Intuitively, through this procedure, we recover the lower singular values only after the input matrix is sufficiently denoised, i.e. sparse corruptions at the desired level have been removed. Figure~\ref{fig:path} shows a pictorial representation of the alternating projections in different stages.

{\bf Parameters:} As can be seen, the only real parameter to the algorithm is $\beta$, used in thresholding, which represents ``spikiness'' of $\Lo$. That is if the user expects $\Lo$ to be ``spiky'' and the sparse part to be heavily diffused, then higher value of $\beta$ can be provided. In our implementation, we found that selecting $\beta$ aggressively helped speed up recovery of our algorithm. In particular, we selected $\beta=1/\sqrt{n}$.

{\bf Complexity:} The complexity of each iteration within a single stage is $O(kmn)$, since it involves
calculating the rank-$k$ approximation\footnote{Note that we only require a rank-$k$ approximation of the matrix rather than the actual singular vectors. Thus, the computational complexity has no dependence on the gap between the singular values.} of an $m\times n$ matrix (done e.g. via vanilla PCA). The number of
iterations in each stage is $\order{\log\left(1/\epsilon\right)}$ and there are at most $r$ stages.
Thus the overall complexity of the entire algorithm is then $O(r^2 mn\log(1/\epsilon))$.
This is drastically lower than the best known bound of $\order{m^2n/\epsilon}$ on the number of iterations
required by convex methods, and just a factor $r$ away from the complexity of vanilla PCA.

\floatname{algorithm}{Algorithm}

\begin{algorithm}[t!]
  \caption{$(\widehat{L},\ \widehat{S})=$ AltProj$(\M, \epsilon, r,\beta)$: Non-convex Alternating Projections based  Robust PCA}\label{algo:sap}
  \begin{algorithmic}[1]
    \STATE {\bf Input}: Matrix $\M\in \R^{m\times n}$, convergence criterion $\epsilon$, target rank $r$, thresholding parameter $\beta$. 
    \STATE $P_k(A)$ denotes the best rank-$k$ approximation of matrix $A$. $HT_\zeta(A)$ denotes hard-thresholding, i.e. $(HT_{\zeta}(A))_{ij}=A_{ij}$ if $|A_{ij}|\geq \zeta$ and $0$ otherwise.
    \STATE Set initial threshold $\zeta_0 \, \leftarrow \, \beta \sigma_1(M)$.
    \STATE $L^{(0)}=0, S^{(0)}=HT_{\zeta_0}(M-L^{(0)})$
    \FOR{Stage $k=1$ to $r$ }
    \FOR{Iteration $t=0$ to $T=10 \log \left(n \beta \twonorm{\M-\S^{(0)}}/\epsilon\right)$}
    \STATE Set threshold $\zeta$ as
    \begin{equation}\zeta ~ =  ~ \beta\, \left(\sigma_{k+1}(\M-\St) +\left(\frac{1}{2}\right)^t \sigma_k(\M-\St)  \right)\label{eqn:threshold}  \end{equation}
    \STATE $L^{(t+1)}=P_k(M-S^{(t)})$
    \STATE $S^{(t+1)}=HT_{\zeta}(M-L^{(t+1)})$
    \ENDFOR
    \IF{$\beta \sigma_{k+1}(\Ltn) < \frac{\epsilon}{2n}$}
    \STATE {\bf Return: }$\L^{(T)},\S^{(T)}$  \hspace{0.5em}{\em/* Return rank-$k$ estimate if remaining part has small norm */}
    \ELSE
    \STATE $S^{(0)}=S^{(T)}$  \hspace{5em}{\em/* Continue to the next stage */}
    \ENDIF
    \ENDFOR
    \STATE {\bf Return: }$L^{(T)}, S^{(T)}$
  \end{algorithmic}
  \label{algo:alg1}
\end{algorithm}

\renewcommand{\u}{\bm{u}}
\section{Analysis}\label{sec:analysis}
In this section, we present our main result on the correctness of \ncralgo. 
We assume the following conditions:
\begin{enumerate}
\item[(L1)] Rank of $\Lo$ is at most $r$.
\item[(L2)] $\Lo$ is $\mu$-incoherent, i.e., if $\Lo=\Uo\Sigmao\Vot$ is the SVD of $\Lo$, then  $\|(\Uo)^i\|_2\leq \frac{\mu\sqrt{r}}{\sqrt{m}},$ $\forall1\leq i\leq m$ and $\|(\Vo)^i\|_2\leq \frac{\mu\sqrt{r}}{\sqrt{n}}$, $\forall 1\leq i\leq n$,
where $(\Uo)^i$ and $(\Vo)^i$ denote the $i^{\textrm{th}}$ rows of $\Uo$ and $\Vo$ respectively.
\end{enumerate}
\begin{enumerate}
\item[(S1)] Each row and column of $S$ have  at most $\alpha $ fraction of non-zero entries such that $\alpha\leq \frac{1}{512 \mu^2 r}$.
\end{enumerate}
Note that  in general,  it is not possible to have a unique recovery of  low-rank and sparse components. For example, if the input matrix $M $ is both sparse and low rank, then there is no unique decomposition (e.g. $M= \e_1\e_1^\top$). The above conditions ensure uniqueness of the matrix decomposition problem.

Additionally, we set the parameter $\beta$  in Algorithm~\ref{algo:alg1} be set as $\beta=\frac{4\mu^2 r}{\sqrt{mn}}$.



We now establish that our proposed algorithm recovers the low rank and sparse components under the above conditions.
\begin{theorem}[Noiseless Recovery]
Under conditions $(L1)$, $(L2)$ and $\So$, and choice of $\beta$ as above,    
the outputs $\widehat{\L}$ and $\widehat{\S}$ of Algorithm~\ref{algo:alg1} satisfy:
\begin{align*}
\frob{\widehat{\L} - \Lo} \leq \epsilon, \infnorm{\widehat{\S} - \So} \leq \frac{\epsilon}{\sqrt{mn}}, \mbox{ and }\ \ \supp{\widehat{\S}} \subseteq \supp{\So}.
\end{align*}
\label{thm:main}
\end{theorem}
\textbf{Remark (tight recovery conditions): }Our result is tight up to constants, in terms of allowable sparsity level
under the deterministic sparsity model. In other words, if we exceed the sparsity limit imposed in S1, it is possible to construct instances where there is no unique decomposition\footnote{For instance, consider the $n\times n$ matrix which has $r$ copies of the all ones matrix, each of
size $\frac{n}{r}$, placed across the diagonal.
We see that this matrix has rank $r$ and is incoherent with parameter
$\mu = 1$. Note that a fraction of $\alpha = \order{1/r}$ sparse perturbations suffice to erase one of these blocks
making it impossible to recover the matrix.}. 
Our conditions L1, L2 and S1 also match the conditions required by the convex method for recovery, as established in~\cite{hsu2011robust}.\\[0pt]
\textbf{Remark (convergence rate): } 
Our method has a linear rate of convergence, i.e. $O(\log(1/\epsilon))$ to achieve an  error of $\epsilon$, and
hence we provide a strongly polynomial method for robust PCA. In contrast, the best known bound for
convex methods for robust PCA is $O(1/\epsilon)$ iterations to converge to an $\epsilon$-approximate solution.

Theorem~\ref{thm:main} provides recovery guarantees assuming that $\Lo$ is exactly rank-$r$. However, in several real-world
scenarios, $\Lo$ can be nearly rank-$r$.
Our algorithm can handle such situations, where $\M=\Lo+\No+\So$, with $\No$ being an additive noise.
Theorem~\ref{thm:main} is  a special case of the following theorem which provides recovery guarantees
when $\No$ has small $\ell_{\infty}$ norm. 


\begin{theorem}[Noisy Recovery]
Under conditions $(L1)$, $(L2)$ and $\So$, and choice of $\beta$ as in Theorem~\ref{thm:main},  when the noise  $\infnorm{\No}\leq \frac{\sigma_r(\Lo)}{100 n}$,
the outputs $\widehat{\L}, \widehat{\S}$ of Algorithm~\ref{algo:alg1} satisfy:
\begin{align*}
\frob{\widehat{\L} - \Lo} &\leq \epsilon + 2 \mu^2r\left(7 \twonorm{\No} + \frac{8 \sqrt{mn}}{\sqrt{r}} \infnorm{\No}\right), \\
\infnorm{\widehat{\S} - \So} &\leq \frac{\epsilon}{\sqrt{mn}} + \frac{2 \mu^2r}{\sqrt{mn}}\left(7 \twonorm{\No} + \frac{8 \sqrt{mn}}{\sqrt{r}} \infnorm{\No}\right), \mbox{ and }
\supp{\widehat{\S}} \subseteq \supp{\So}.
\end{align*}
\label{thm:noise}
\end{theorem}
\subsection{Proof Sketch}\label{sec:thmmain-outline}
We now present the key steps in the proof of Theorem~\ref{thm:main}. A detailed proof is provided in the appendix. 

\textbf{Step I: Reduce to the symmetric case, while maintaining incoherence of $\bm{\Lo}$ and sparsity of $\bm{\So}$}.
Using standard symmetrization arguments, we can reduce the problem to the symmetric case, where all the matrices
involved are symmetric. See appendix for details on this step.

\textbf{Step II: Show decay in $\|L-L^*\|_\infty$ after projection onto the set of rank-$k$ matrices. } The $t$-th iterate $\Ltn$ of the $k$-th stage is given by $\Ltn=P_k(\Lo+\So-\St)$. Hence, $\Ltn$ is obtained by using the top principal components of a perturbation of $\Lo$ given by $\Lo+(\So-\St)$. The key step in our analysis is to show that when an incoherent and low-rank $\Lo$ is perturbed by a sparse matrix $\So-\St$, then $\|\Ltn-\Lo\|_\infty$ is small and is much smaller than $|\So-\St|_\infty$. The following lemma formalizes the intuition; see the appendix for a detailed proof. 
\begin{lemma}\label{lem:dec}
Let $\Lo, \So$ be symmetric and satisfy the assumptions of Theorem~\ref{thm:main} and let $\St$ and $\Lt$ be the
$t^{\textrm{th}}$ iterates of the $k^{\textrm{th}}$ stage of Algorithm~\ref{algo:alg1}. Let $\sigma_1^*, \dots, \sigma_n^*$ be the eigenvalues
of $\Lo$, s.t., $|\sigma_1^*|\geq \dots \geq |\sigma_r^*|$.
Then, the following holds: 
\begin{align*}
\infnorm{\Ltn-\Lo} \leq \frac{2\mu^2 r}{n}&\left(\abs{\sigma_{k+1}^*}+\left(\frac{1}{2}\right)^{t}\abs{\sigma_k^*}
 \right), \\\infnorm{\So-\Stn} \leq \frac{8\mu^2 r}{n}&\left(\abs{\sigma_{k+1}^*}+\left(\frac{1}{2}\right)^{t}\abs{\sigma_k^*}
	\right), \mbox{ and } \supp{\Stn} \subseteq \supp{\So}.
\end{align*}
Moreover, the outputs $\widehat{\L}$ and $\widehat{\S}$ of Algorithm~\ref{algo:alg1} satisfy:
\begin{align*}
\frob{\widehat{\L} - \Lo} \leq \epsilon, \ \ \infnorm{\widehat{\S} - \So} \leq \frac{\epsilon}{{n}}, \mbox{ and }\ \ \supp{\widehat{\S}} \subseteq \supp{\So}.
\end{align*}
\end{lemma}
\textbf{Step III: Show decay in $\|S-S^*\|_\infty$ after projection onto the set of sparse matrices.} We next show that if $\|\Ltn-\Lo\|_\infty$ is much smaller than $\|\St-\So\|_\infty$ then the iterate $\Stn$ also has a much smaller error (w.r.t. $\So$) than $\St$. The above given lemma formally provides the error bound. 

\textbf{Step IV: Recurse the argument.}
We have now reduced the $\ell_{\infty}$ norm of the sparse part by a factor of half, while maintaining
its sparsity. We can now go back to steps II and III and repeat the arguments for subsequent iterations.

\section{Experiments}\vspace*{-5pt}
We now present an empirical study of our AltProj method. The goal of this study is two-fold: a) establish that our method indeed recovers the low-rank and sparse part exactly, without significant parameter tuning, b) demonstrate that AltProj is significantly faster than Conv-RPCA (see \eqref{eq:conv}); we solve Conv-RPCA using the IALM method \cite{CandesLMW11}, a state-of-the-art solver \cite{lin2010augmented}.
We implemented our method in Matlab and used a Matlab
implementation of the IALM method by \cite{lin2010augmented}.

We consider both synthetic experiments and experiments on real data involving   the
problem of foreground-background separation in a video.  Each of our results for synthetic datasets is averaged
over $5$ runs.

{\em Parameter Setting}: Our pseudo-code (Algorithm~\ref{algo:alg1}) prescribes the threshold $\zeta$ in Step 4, which depends on the knowledge of the singular values of the low rank component $L^*$. Instead, in the experiments, we set the threshold at the $(t+1)$-th step of $k$-th stage as $\zeta=\frac{\mu\sigma_{k+1}(\M-\St)}{\sqrt{n}}$.  For synthetic experiments, we employ the $\mu$ used for data generation,   and for real-world datasets, we tune    $\mu$ through cross-validation. We found that the above thresholding  provides exact recovery while speeding up the computation significantly. We would also like to note that \cite{CandesLMW11} sets the regularization parameter $\lambda$ in Conv-RPCA \eqref{eq:conv} as $1/\sqrt{n}$ (assuming $m \leq n$). However, we found that for problems with large incoherence such a parameter setting {\em does not} provide exact recovery. Instead, we set $\lambda=\mu/\sqrt{n}$ in our experiments.

\begin{figure}[t]
\centering
\begin{tabular}{cccc}
\hspace*{-10pt}\includegraphics[width=0.25\textwidth]{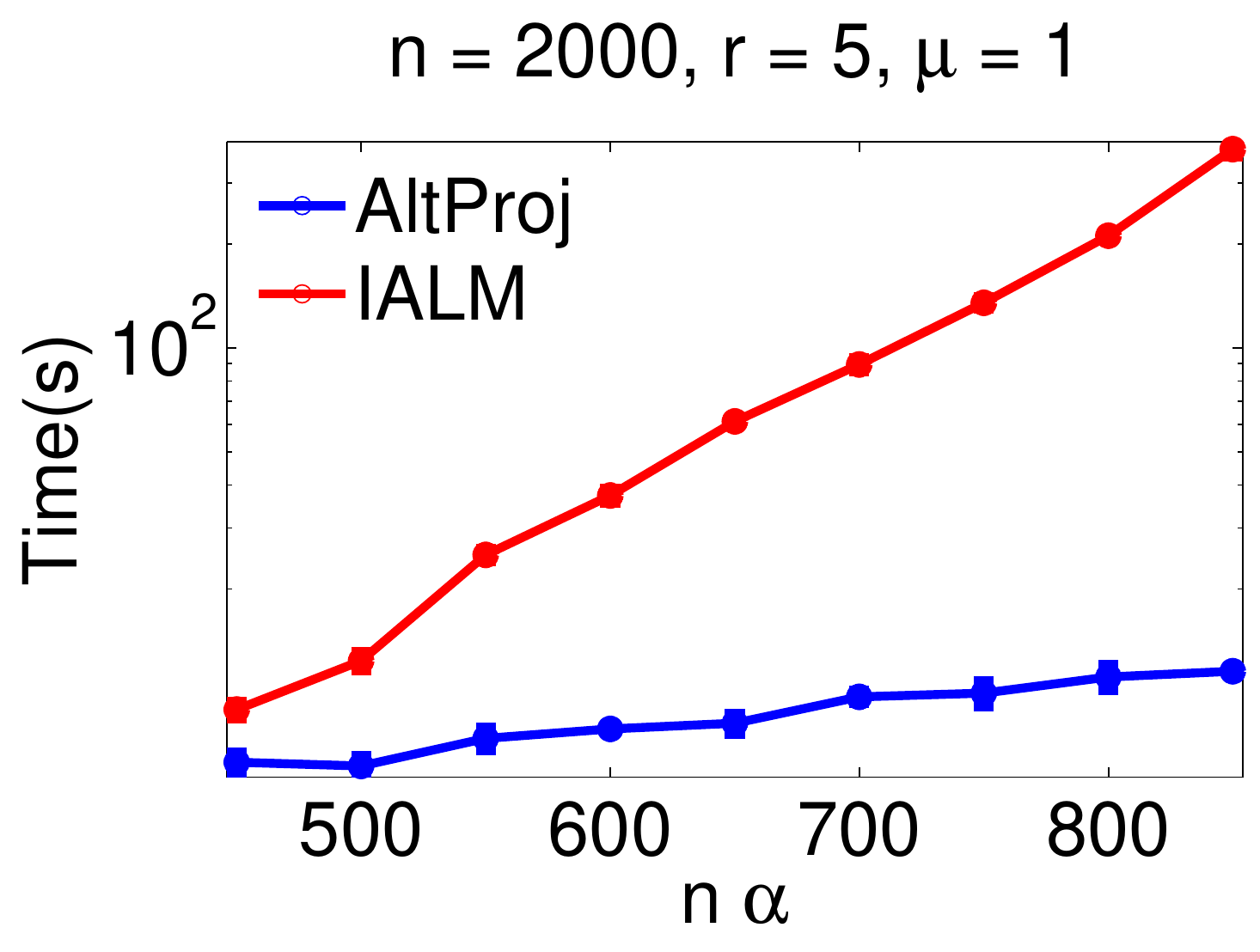}&
\hspace*{-10pt}\includegraphics[width=0.25\textwidth]{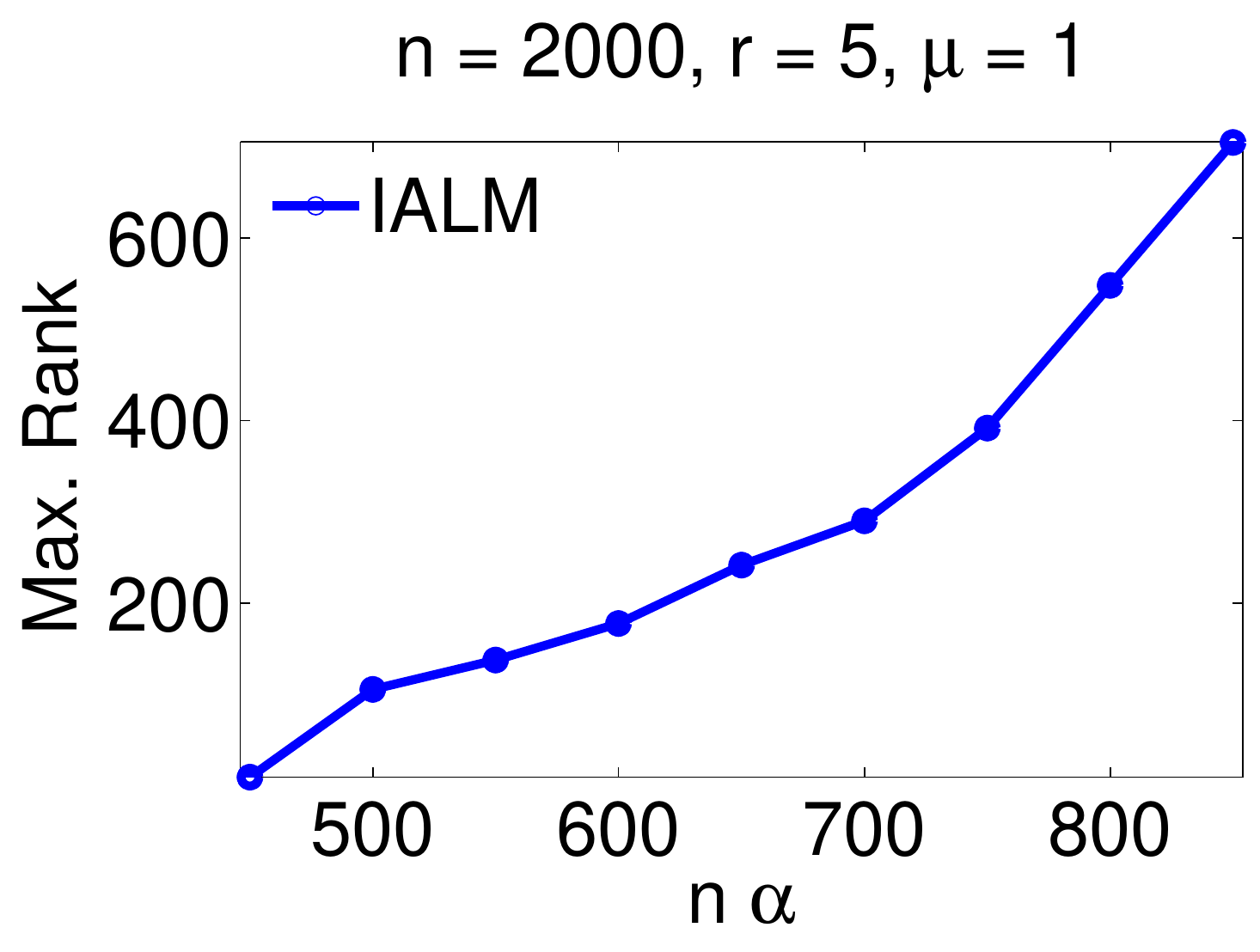}&
\hspace*{-10pt}\includegraphics[width=0.25\textwidth]{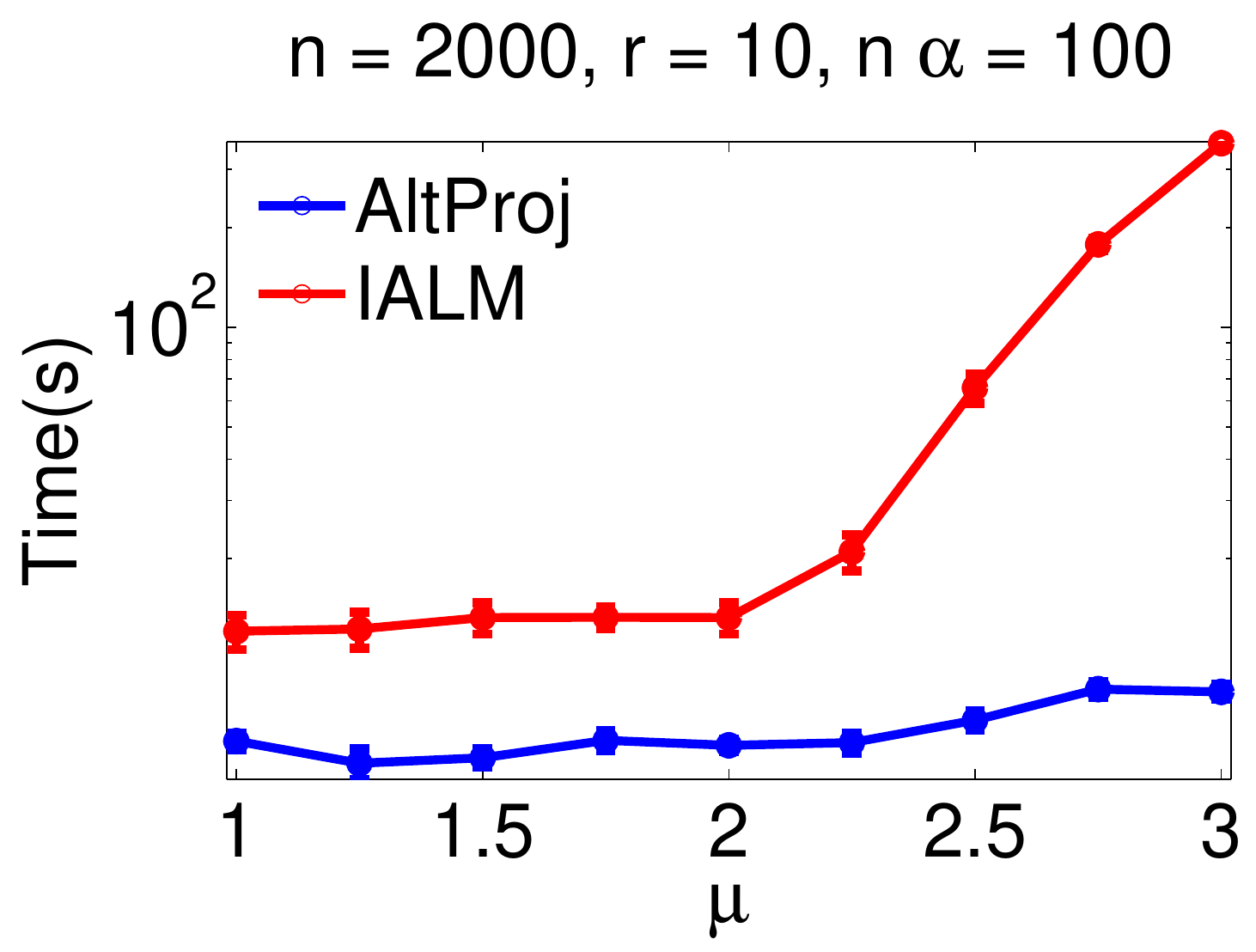}&
\hspace*{-10pt}\includegraphics[width=0.25\textwidth]{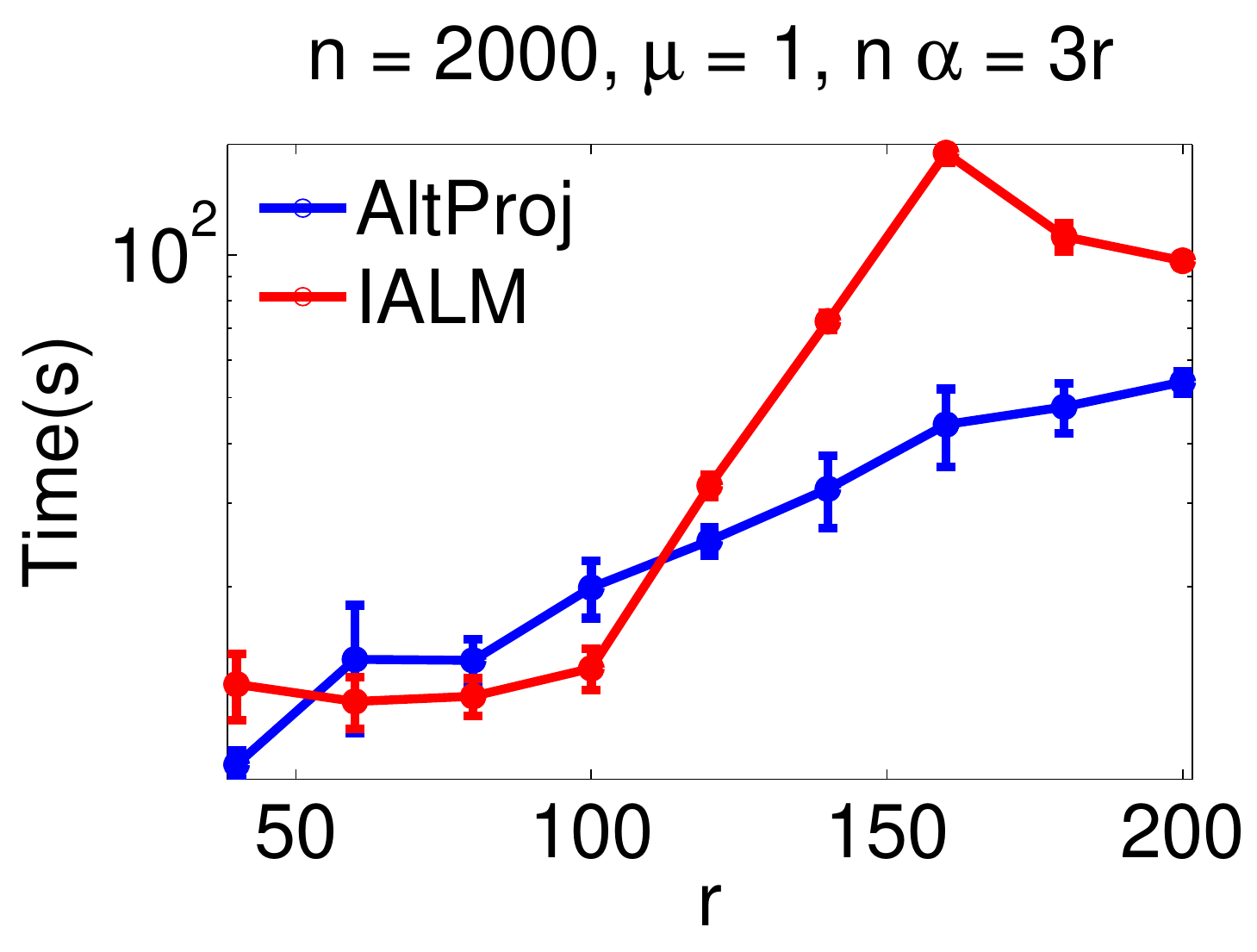}\\[-3pt]
(a)&(b)&(c)&(d)\\
\end{tabular}\vspace*{-8pt}
\caption{Comparison of AltProj and IALM on synthetic datasets. (a) Running time of AltProj and IALM with varying $\alpha$. 
(b) Maximum rank of the intermediate iterates of IALM. (c) Running time of AltProj and IALM with varying $\mu$. 
(d) Running time of AltProj and IALM with varying $r$.}
\label{fig:synthetic_plots_1}\vspace*{-15pt}
\end{figure}

{\bf Synthetic datasets:} Following the experimental setup of \cite{CandesLMW11}, the low-rank part $\Lo=UV^T$ is generated using normally distributed $U\in \R^{m\times r}$, $V\in \R^{n\times r}$. Similarly, $supp(\So)$ is generated by sampling a uniformly random subset of $[m]\times [n]$  with size $\|\So\|_0$ and each non-zero $\So_{ij}$ is drawn i.i.d. from the uniform distribution over $[ r / (2\sqrt{mn}), r / \sqrt{mn} ]$. For increasing incoherence of $\Lo$, we randomly zero-out rows of $U, V$ and then re-normalize them.

There are three key problem parameters for RPCA with a fixed matrix size: a) sparsity of $\So$, b) incoherence of $\Lo$, c) rank of $\Lo$. We investigate performance of both AltProj and IALM by varying each of the three parameters while fixing the others. In our plots (see Figure~\ref{fig:synthetic_plots_1}), we report computational time required by each of the two methods for decomposing $\M$ into $\L+\S$ up to a relative error ($\|\M-\L-\S\|_F/\|\M\|_F$) of $10^{-3}$. Figure~\ref{fig:synthetic_plots_1} shows that AltProj scales significantly better than IALM for increasingly dense $\So$. We attribute this observation to the fact that as $\|\So\|_0$ increases, the problem is ``harder'' and the intermediate iterates of IALM have ranks significantly larger than $r$. Our intuition is confirmed by Figure~\ref{fig:synthetic_plots_1} (b), which shows that when density ($\alpha$) of $\So$ is  $0.4$ then the intermediate iterates of IALM can be of rank over $500$ while the rank of $\Lo$ is only $5$. We observe a similar trend for the other parameters, i.e., AltProj scales significantly better than IALM with increasing incoherence parameter $\mu$ (Figure~\ref{fig:synthetic_plots_1} (c)) and increasing rank (Figure~\ref{fig:synthetic_plots_1} (d)). See Appendix~\ref{app:exps} for additional plots.

\begin{figure}[t]
\centering
\begin{tabular}{cccc}
\hspace*{-7pt}\includegraphics[width=0.25\textwidth]{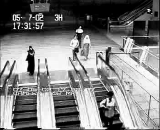}&
\hspace*{-10pt}\includegraphics[width=0.25\textwidth]{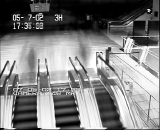}&
\hspace*{-10pt}\includegraphics[width=0.25\textwidth]{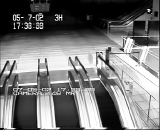}&
\hspace*{-10pt}\includegraphics[width=0.25\textwidth]{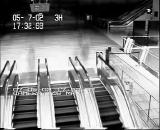}\\
(a)&(b)&(c)&(d)
\end{tabular}\vspace*{-5pt}
\caption{Foreground-background separation in the {\em Escalator} video. (a): Original image frame. (b): Best rank-$10$ approximation; time taken is $3.1s$. (c): Low-rank frame obtained using AltProj; time taken is $63.2s$. (d): Low-rank frame obtained using IALM; time taken is $1688.9s$.}
\label{fig:escalator}
\end{figure}
{\bf Real-world datasets: }Next, we apply our method to the problem of foreground-background (F-B) separation in a video \cite{li2004statistical}. The observed matrix $\M$ is formed by vectorizing each frame and stacking them column-wise. Intuitively, the background in a video is the static part and hence forms a low-rank component while the foreground is a dynamic but sparse perturbation.

Here, we used two benchmark datasets named {\em Escalator} and {\em Restaurant} dataset. The {\em Escalator} dataset has $3417$ frames at a resolution of $160 \times 130$. We first applied the standard PCA method for extracting low-rank part. Figure~\ref{fig:escalator} (b) shows the extracted background from the video. There are several artifacts (shadows of people near the escalator) that are not desirable. In contrast, both IALM and AltProj obtain significantly better F-B separation (see Figure~\ref{fig:escalator}(c), (d)). Interestingly, AltProj removes the steps of the escalator which are moving and arguably are part of the dynamic foreground, while IALM keeps the steps in the background part. Also, our method is significantly faster, i.e., our method, which takes $63.2s$ is about $26$ times faster than IALM, which takes $1688.9s$.

\emph{Restaurant dataset: }Figure~\ref{fig:restaurant} shows the comparison of AltProj and IALM on a subset of the ``Restaurant'' dataset where we consider the last $2055$ frames at a resolution of $120 \times 160$. AltProj was around $19$ times faster than IALM. Moreover, visually, the background extraction seems to be of better quality (for example, notice the blur near top corner counter in the IALM solution). Plot(b) shows the PCA solution and that also suffers from a similar blur at the top corner of the image, while the background frame extracted by AltProj does not have any noticeable artifacts.
\begin{figure}[t]
\centering
\begin{tabular}{cccc}
\hspace*{-7pt}\includegraphics[width=0.25\textwidth]{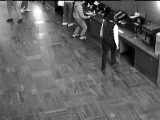}&
\hspace*{-10pt}\includegraphics[width=0.25\textwidth]{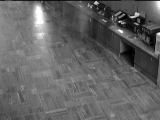}&
\hspace*{-10pt}\includegraphics[width=0.25\textwidth]{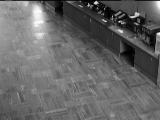}&
\hspace*{-10pt}\includegraphics[width=0.25\textwidth]{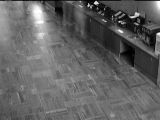}\\
(a)&(b)&(c)&(d)
\end{tabular}
\caption{Foreground-background separation in the {\em Restaurant} video. (a): Original  frame from the video. (b):  Best rank-$10$ approximation (using PCA) of the original frame; $2.8s$ were required to compute the solution (c): Low-rank part obtained using AltProj; computational time required by AltProj was $34.9s$. (d): Low-rank part obtained using IALM; $693.2s$ required by IALM to compute the low-rank+sparse decomposition.}
\label{fig:restaurant}\vspace*{-10pt}
\end{figure}



\vspace*{-8pt}\section{Conclusion}\vspace*{-8pt}
In this work, we proposed a  non-convex method for robust PCA, which consists of alternating projections on to low rank and sparse matrices. We established global convergence of our method under conditions which match those for convex methods. At the same time, our method has much faster running times, and has superior experimental performance. This work opens up a number of interesting questions for future investigation. While we match the convex methods, under the deterministic sparsity model, studying the random sparsity model is of interest. Our noisy recovery results assume deterministic noise; improving the results under random noise needs to be investigated. There are many decomposition problems beyond the robust PCA setting, e.g. structured sparsity models, robust tensor PCA problem, and so on. It is interesting to see if we can establish global convergence for non-convex methods in these settings.  \vspace*{-8pt}
\section*{Acknowledgements}\vspace*{-8pt}
AA and UN would like to acknowledge NSF grant CCF-1219234, ONR N00014-14-1-0665, and Microsoft faculty fellowship. SS would like to acknowledge NSF grants 1302435, 0954059, 1017525 and DTRA grant  HDTRA1-13-1-0024. PJ would like to acknowledge Nikhil Srivastava and Deeparnab Chakrabarty for several insightful discussions during the course of the project.  

\bibliographystyle{alpha}

\newcommand{\etalchar}[1]{$^{#1}$}

\clearpage
\appendix

\section{Proof of Theorem~\ref{thm:main}}
\label{app:proofs}
We will start with some preliminary lemmas.
The first lemma is the well known Weyl's inequality in the matrix setting\cite{bhatia}.
\begin{lemma}\label{lem:weyl-perturbation}
Suppose $B = A + E$ be an $n \times n$ matrix. Let $\lambda_1,\cdots,\lambda_n$ and $\sigma_1,\cdots,\sigma_n$ be the eigenvalues of
$B$ and $A$ respectively such that $\lambda_1 \geq \cdots \geq \lambda_n$ and $\sigma_1
\geq \cdots \geq {\sigma_n}$. Then we have:
\begin{align*}
  \abs{\lambda_i - \sigma_i} \leq \twonorm{E} \; \forall \; i \in [n].
\end{align*}
\end{lemma}

The following lemma is the Davis-Kahan theorem\cite{bhatia},
specialized for rank-$1$ matrices.
\begin{lemma}\label{lem:daviskahan}
Suppose $\mB = \A + \E$. Let $\A = \uo \trans{(\uo)}$ be a rank-$1$ matrix with unit spectral norm.
Suppose further that $\twonorm{\E}< \frac{1}{2}$. Then, we have:
\begin{align*}
  \abs{\lambda - 1} &< \twonorm{\E}, \mbox{ and }\\
  \abs{\iprod{\u}{\uo}^2 - 1} &< 4 \twonorm{\E},
\end{align*}
where $\lambda$ and $\u$ are the top eigenvalue eigenvector pair of $\mB$.
\end{lemma}

As outlined in Section~\ref{sec:thmmain-outline} (and formalized in the proof of Theorem~\ref{thm:main}), it is sufficient to prove the correctness of Algorithm~\ref{algo:alg1}
for the case of symmetric matrices. So, most of the lemmas we prove in this section assume that the matrices are
symmetric.
\begin{lemma}
  \label{lem:spec1}
Let $\S\in \R^{n \times n}$ satisfy assumption $(S1)$.
Then, $\twonorm{S} \leq \alpha n \infnorm{S}$.
\end{lemma}

\begin{proof}[Proof of Lemma~\ref{lem:spec1}]
  Let $x, y$ be unit vectors such that $\|S\|_2=x^T S y=\sum_{ij}x_i y_j S_{ij}$. Then, using $a\cdot b \leq (a^2+b^2)/2$,
  we have:
  \begin{align}
    \twonorm{S} \leq \frac{1}{2}\sum_{ij} (x_i^2 + y_j^2)S_{ij}\leq \frac{1}{2} (\alpha n \|S\|_\infty+\alpha n \|S\|_\infty),
  \end{align}
where the last inequality follows from the fact that $S$ has at most $\alpha n$ non-zeros per row and per column. 
\end{proof}

\begin{lemma}\label{lem:inf}
  Let $S\in \R^{n\times n}$ satisfy assumption ($S1$). Also, let $U\in \R^{n\times r}$ be a $\mu$-incoherent
  orthogonal matrix, i.e., $\max_i \twonorm{\trans{\e_i} U} \leq \frac{\mu\sqrt{r}}{\sqrt{n}}$, where $\e_i$ stands
  for the $i^{th}$ standard basis vector.
  Then, $\forall p\geq 0$, the following holds: $$\max_i \twonorm{\trans{\e_i} S^p U} \leq \frac{\mu\sqrt{r}}{\sqrt{n}}(\alpha\cdot n\cdot \infnorm{S})^p.$$
\end{lemma}

\begin{proof}[Proof of Lemma~\ref{lem:inf}]
  We prove the lemma using mathematical induction. 

Base Case ($p=0$): This is just a restatement of the incoherence of $U$. 

Induction step: We have:
\begin{align*}
  \twonorm{\trans{\e_i} (S)^{p+1} U}^2 &= \|\trans{\e_i} S (S^p U)\|_2^2 = \sum_\ell (\trans{\e_i} S (S^p U) \e_\ell)^2
  =\sum_\ell (\sum_j S_{ij}\trans{\e_j} (S^p U) \e_\ell)^2\\
&=\sum_{j_1 j_2}S_{ij_1}S_{ij_2}\sum_{\ell}(\trans{\e_{j_1}}(S^p U) \e_\ell)(\trans{\e_\ell}\trans{(S^p U)} \e_{j_2})\\
&\stackrel{\zeta_1}{\leq}\sum_{j_1 j_2}S_{ij_1}S_{ij_2}(\trans{\e_{j_1}}(S^p U) \trans{(S^p U)} \e_{j_2})
\leq\sum_{j_1 j_2}S_{ij_1}S_{ij_2} \|\e_{j_1}^T(S^p U)\|_2 \|\trans{\e_{j_2}}(S^p U)\|_2\\
&\stackrel{\zeta_2}{\leq}\frac{\mu^2 r}{n} (\alpha\cdot n\cdot \|S\|_\infty)^{2p}, 
\end{align*}
where $\zeta_1$ follows by $\sum_{\ell=1}^t \e_\ell \trans{\e_\ell}=I$, and $\zeta_2$ follows from assumption
$(S1)$ on $S$ and from the inductive hypothesis on $\twonorm{\trans{\e_i} S^p U}$.
\end{proof}

In what follows, we prove a number of lemmas concerning the structure of $\Lt$ and $\Et\defas\So-\St$.
The following lemma shows that the threshold in \eqref{eqn:threshold} is close to
that with $\M-\St$ replaced by $\Lo$.
\begin{lemma}\label{lem:threshold_est}
Let $\Lo, \So$ be symmetric and satisfy the assumptions of Theorem~\ref{thm:main} and let $\St$ be the
$t^{\textrm{th}}$ iterate of the $k^{\textrm{th}}$ stage of Algorithm~\ref{algo:alg1}. Let $\sigma_1^*, \dots,
\sigma_r^*$ be the eigenvalues of $\Lo$, such that $|\sigma_1^*|\geq \dots \geq |\sigma_r^*|$ and $\lambda_1,
\cdots,\lambda_n$ be the eigenvalues of $\M-\St$ such that $\abs{\lambda_1} \geq \cdots \geq \abs{\lambda_n}$.
Recall that $\Et \defas \So-\St$. Suppose further that
\begin{enumerate}
  \item	$\infnorm{\Et} \leq \frac{8\mu^2 r}{n}\left(|\sigma_{k+1}^*|+\left(\frac{1}{2}\right)^{t-1}|\sigma_k^*|\right)$, and
  \item	$\supp{\Et} \subseteq \supp{\So}$.
\end{enumerate}
Then,
\begin{align}
  \label{eq:supp_1}
  \frac{7}{8}\left(|\sigma_{k+1}^*|+\left(\frac{1}{2}\right)^{t}|\sigma_k^*|\right)
  \leq \left(|\lambda_{k+1}|+\left(\frac{1}{2}\right)^{t}|\lambda_k|\right)
  \leq \frac{9}{8} \left(|\sigma_{k+1}^*|+\left(\frac{1}{2}\right)^{t}|\sigma_k^*|\right).
\end{align}
\end{lemma}
\begin{proof}
Note that $\M-\St=\Lo+\Et$. Now, using Lemmas~\ref{lem:weyl-perturbation} and~\ref{lem:spec1}, we have:
\begin{align*}
\abs{\lambda_{k+1}-\sigma_{k+1}^*} \leq \twonorm{\Et} \leq \alpha n \infnorm{\Et} \leq 8\mu^2 r \alpha \gamma_t,
\end{align*}
where $\gamma_t \defas \left(|\sigma_{k+1}^*|+\left(\frac{1}{2}\right)^{t-1}|\sigma_k^*|\right)$.
That is,
$\abs{|\lambda_{k+1}|-|\sigma_{k+1}^*|}\leq  8\mu^2 r \alpha \gamma_t$.
Similarly, $\abs{|\lambda_{k}|-|\sigma_{k}^*|} \leq  8\mu^2 r \alpha \gamma_t$.
So we have:
\begin{align*}
\abs{ \left(|\lambda_{k+1}|+\left(\frac{1}{2}\right)^{t}|\lambda_k|\right) - \left(|\sigma_{k+1}^*|+\left(\frac{1}{2}\right)^{t}|\sigma_k^*|\right)} &\leq 8\mu^2 r \alpha \gamma_t \left(1+\left(\frac{1}{2}\right)^t\right) \\
&\leq 16 \mu^2 r \alpha \gamma_t \\
&\leq \frac{1}{8}\left(|\sigma_{k+1}^*|+\left(\frac{1}{2}\right)^{t}|\sigma_k^*|\right),
\end{align*}
where the last inequality follows from the bound $\alpha\leq \frac{1}{512\mu^2 r}$. 
\end{proof}

The following lemma shows that under the same assumptions as in Lemma~\ref{lem:threshold_est},
we can obtain a bound on the $\ell_{\infty}$ norm of $\Ltn-\Lo$.
This is the most crucial step in our analysis since we bound $\ell_{\infty}$ norm of errors which are quite hard
to obtain.
\begin{lemma}\label{lem:dec1}
Assume the notation of Lemma~\ref{lem:threshold_est}. Also, let $\Lt, \St$ be the $t^{\textrm{th}}$ iterates of
$k^{\textrm{th}}$ stage of Algorithm~\ref{algo:alg1} and $\Ltn, \Stn$ be the $(t+1)^{\textrm{th}}$ iterates of the same
stage. Also, recall that $\Et \defas \So-\St$ and $\Etn \defas \So-\Stn$.
Suppose further that
\begin{enumerate}
  \item	$\infnorm{\Et} \leq \frac{8\mu^2 r}{n}\left(|\sigma_{k+1}^*|+\left(\frac{1}{2}\right)^{t-1}|\sigma_k^*| \right)$, and
  \item	$\supp{\Et}\subseteq \supp{\So}$.
\end{enumerate}
Then, we have:
  $$\infnorm{\Ltn-\Lo} \leq \frac{2\mu^2 r}{n}\left(|\sigma_{k+1}^*|+\left(\frac{1}{2}\right)^{t}|\sigma_k^*| \right).$$
\end{lemma}
\begin{proof}
  Let $\Ltn=P_k(\M-\St)=U\Lambda \trans{U}$ be the eigenvalue decomposition of $\Ltn$. Also, recall that
  $\M-\St=\Lo+\Et$. Then, for every eigenvector $\u_i$ of $\Ltn$, we have
  \begin{align}
    \left(\Lo+\Et\right)\u_i&=\lambda_i \u_i, \nonumber\\
    \left(I-\frac{\Et}{\lambda_i}\right)\u_i&=\frac{1}{\lambda_i} \Lo \u_i,\nonumber\\
\u_i&=\left(I-\frac{\Et}{\lambda_i}\right)^{-1}\frac{\Lo\u_i}{\lambda_i}\nonumber\\
&=\left(I+\frac{\Et}{\lambda_i}+\left(\frac{\Et}{\lambda_i}\right)^2+\dots\right)\frac{\Lo \u_i}{\lambda_i}. 
  \end{align}
Note that we used Lemmas~\ref{lem:weyl-perturbation} and~\ref{lem:spec1} to guarantee the existence of
$\inv{\left(I - \frac{\Et}{\lambda_i}\right)}$.
Hence, 
\begin{align*}
  U\Lambda \trans{U} - \Lo= \left(\Lo U \Lambda^{-1} \trans{U} \Lo -\Lo\right)
  + \sum_{p+q\geq 1} \left(\Et \right)^p \Lo U \Lambda^{-(p+q+1)} \trans{U} \Lo \left(\Et \right)^q. 
\end{align*}
By triangle inequality, we have
\begin{multline}
  \label{eq:dec1_n}
    \infnorm{U\Lambda \trans{U} - \Lo} \leq \infnorm{\Lo U \Lambda^{-1} \trans{U} \Lo -\Lo}\\
    + \sum_{p+q\geq 1} \infnorm{\left(\Et\right)^p \Lo U \Lambda^{-(p+q+1)} \trans{U} \Lo \left(\Et\right)^q}. 
\end{multline}
We now bound the two terms on the right hand side above.

We note that, 
\begin{align}
  &\infnorm{\Lo U \Lambda^{-1} \trans{U} \Lo -\Lo} \nonumber \\
  &=\max_{ij}\trans{\e_i}\left(\Uo\Sigmao\trans{(\Uo)}U \Lambda^{-1} \trans{U}\Uo\Sigmao\trans{(\Uo)}-\Uo\Sigmao\trans{(\Uo)}\right)\e_j\nonumber\\
  &=\max_{ij}\trans{\e_i}\Uo\left(\Sigmao\trans{(\Uo)}U \Lambda^{-1} \trans{U}\Uo\Sigmao - \Sigmao \right) \trans{(\Uo)}\e_j\nonumber\\
&\leq \max_{ij}\|\trans{\e_i}\Uo\|\cdot \|\trans{\e_j}\Uo\|\cdot \|\Uo\Sigmao\trans{(\Uo)}U \Lambda^{-1} \trans{U}\Uo\Sigmao\trans{(\Uo)}-\Uo\Sigmao\trans{(\Uo)}\|_2\nonumber\\
&\leq \frac{\mu^2r}{n}\|\Lo U \Lambda^{-1} \trans{U}\Lo-\Lo\|_2,\label{eq:dec2_n}
\end{align}
where we denote $\Uo\Sigmao\trans{(\Uo)}$ to be the SVD of $\Lo$.
Let $\Lo +\Et=U\Lambda \trans{U} + \widetilde{U} \widetilde{\Lambda} \trans{\widetilde{U}}$ be the eigenvalue decomposition of $\Lo +\Et$.
Note that $\trans{\widetilde{U}}U=0$. Recall that, $U\Lambda \trans{U}=P_k(\M-\St)=P_k(\Lo+\Et)=\Ltn$. Also note that,
\begin{align}
  &\Lo U \Lambda^{-1} \trans{U}\Lo-\Lo \nonumber \\
  & = \left(U\Lambda \trans{U} + \widetilde{U} \widetilde{\Lambda} \trans{\widetilde{U}}-\Et\right)
  U\Lambda^{-1} \trans{U}
\left(U\Lambda \trans{U} + \widetilde{U} \widetilde{\Lambda} \widetilde{U}^T-\Et\right) - \Lo,\nonumber\\
&=\left(U\trans{U}-\left(\Et\right) U\Lambda^{-1} \trans{U}\right)
\left(U\Lambda \trans{U} + \widetilde{U} \widetilde{\Lambda} \widetilde{U}^T-\Et \right) - \Lo,\nonumber\\
&=-U\trans{U}\Et - \Et U\trans{U} -\Et U\Lambda^{-1}\trans{U} \trans{\Et} -\widetilde{U} \widetilde{\Lambda} \trans{\widetilde{U}} +\Et.\label{eq:dec6_n}
\end{align}
Hence, using Lemma~\ref{lem:spec2_noisy}, we have: 
\begin{align}
  \|\Lo U \Lambda^{-1} \trans{U}\Lo-\Lo\|_2 &\leq 3\|\Et\|_2+\frac{\|\Et\|_2^2}{\abs{\lambda_k}}+\abs{\lambda_{k+1}} \nonumber\\
&\leq \abs{\sigma^*_{k+1}} + 5 \twonorm{\Et}.
  \label{eq:dec3_n}
\end{align}

Combining~\eqref{eq:dec2_n} and~\eqref{eq:dec3_n}, we have: 
\begin{align}
  \infnorm{\Lo U \Lambda^{-1} \trans{U} \Lo -\Lo} \leq \frac{\mu^2r}{n}\left(\abs{\sigma^*_{k+1}}+ 5 \twonorm{\Et}\right)
  \label{eqn:firstterm-final_n}
\end{align}

Now, we will bound the $(p,q)^{\textrm{th}}$ term of $\sum_{p+q\geq 1} \infnorm{\left(\Et\right)^p \Lo U \Lambda^{-(p+q+1)} \trans{U} \Lo \left( \Et \right)^q}$: 
\begin{align}
  &\infnorm{(\Et)^p \Lo U \Lambda^{-(p+q+1)} \trans{U} \Lo (\Et)^q} \nonumber \\
  &\qquad \qquad = \max_{ij} \trans{\e_i} \left((\Et)^p \Lo U \Lambda^{-(p+q+1)} \trans{U} \Lo (\Et)^q\right)\e_j,\nonumber\\
&\qquad \qquad \leq  \max_{ij}\twonorm{\trans{\e_i}(\Et)^p\Uo} \twonorm{\trans{\e_j}(\Et)^q\Uo} \twonorm{\Lo U \Lambda^{-(p+q+1)} \trans{U} \Lo},\nonumber\\
&\qquad \qquad \stackrel{\zeta_1}{\leq} \frac{\mu^2r}{n} \left(\alpha n \infnorm{\Et}\right)^p \left(\alpha n  \infnorm{\Et}\right)^q 
\twonorm{\Lo U \Lambda^{-(p+q+1)} \trans{U} \Lo},\label{eq:dec9_n}
\end{align}
where $\zeta_1$ follows from Lemma~\ref{lem:inf} and the incoherence of $\Lo$.
Now, similar to \eqref{eq:dec6_n}, we have: 
\begin{align}
&  \twonorm{\Lo U \Lambda^{-(p+q+1)} \trans{U} \Lo} \nonumber \\
&\qquad = \left\|U\Lambda^{-(p+q-1)}\trans{U}-\Et U\Lambda^{-(p+q)}\trans{U}-
	U\Lambda^{-(p+q)}\trans{U}\Et
 +\Et U\Lambda^{-(p+q+1)}\trans{U} \Et\right\|_2,\nonumber\\
&\qquad \leq \|\Lambda^{-(p+q-1)}\|_2+2\|\Et\|_2 \|\Lambda^{-(p+q)}\|_2+\|\Et\|_2^2 \|\Lambda^{-(p+q+1)}\|_2,\nonumber\\
&\qquad \leq \abs{\lambda_k}^{-(p+q-1)}\left(1+2\frac{\|\Et\|_2}{\abs{\lambda_k}}+\frac{\|\Et\|_2^2}{\lambda_k^2}\right)= \abs{\lambda_k}^{-(p+q-1)}\left(1+\frac{\|\Et\|_2}{\abs{\lambda_k}}\right)^2,\nonumber\\
&\qquad \leq \abs{\lambda_k}^{-(p+q-1)}\left(1+\frac{\twonorm{\Et}}{\abs{\lambda_k}}\right)^2,\nonumber\\
&\qquad \stackrel{\zeta_1}{\leq}  \abs{\lambda_k}^{-(p+q-1)}\left(1+\frac{17 \mu^2 r \alpha \abs{\sigma_k^*}}{(1-17 \mu^2 r \alpha)\abs{\sigma_k^*}}\right)^2
{\leq}  2 \abs{\lambda_k}^{-(p+q-1)},\label{eq:dec10_n} 
\end{align}
where $\zeta_1$ follows from Lemma~\ref{lem:spec2_noisy}.  

Using \eqref{eq:dec9_n}, \eqref{eq:dec10_n}, we have: 
\begin{align}
  \infnorm{(\Et)^p \Lo U \Lambda^{-(p+q+1)} \trans{U} \Lo (\Et)^q}
  \leq 2\alpha \mu^2 r \infnorm{\Et} \left(\frac{\alpha n \infnorm{\Et}}{\abs{\lambda_k}}\right)^{p+q-1}.
\label{eqn:ps-firstterm-final_n}
\end{align}

Using the above bound, and the  assumption on $\infnorm{\Et}$:
\begin{align*}
\infnorm{\Et} \leq \frac{8\mu^2r}{n}\left(\abs{\sigma_{k+1}^*}+
\left(\frac{1}{2}\right)^{t-1}\abs{\sigma_k^*}
 \right) \leq \frac{17 \mu^2 r}{n}\abs{\sigma_k^*},
\end{align*}
we have:
\begin{align}
&\sum_{p+q\geq 1}\infnorm{\left( \Et \right)^p \Lo U \Lambda^{-(p+q+1)} \trans{U} \Lo \left(\Et\right)^q} \nonumber\\
&\qquad \qquad \leq   2 \mu^2 r \alpha \infnorm{\Et}  \sum_{p+q\geq 1}\left(\frac{\alpha n \infnorm{\Et}}{\abs{\lambda_k}}\right)^{p+q-1} \nonumber\\
&\qquad \qquad \leq 2 \mu^2 r \alpha \infnorm{\Et} \left(\frac{1}{1-\frac{17\mu^2 \alpha r}{1-17\mu^2\alpha\cdot r}}\right)^2 \nonumber\\
&\qquad \qquad \leq 2 \mu^2 r \alpha \infnorm{\Et} \left(\frac{1}{1-34\mu^2 r \alpha}\right)^2 \nonumber\\
&\qquad \qquad \leq 4 \mu^2 r \alpha \infnorm{\Et}. \label{eq:dec12_n}
\end{align}

Combining \eqref{eq:dec1_n}, \eqref{eqn:firstterm-final_n}, \eqref{eq:dec12_n}, we have: 
\begin{align*}
  \|U\Lambda \trans{U} - \Lo\|_\infty
	&\stackrel{}{\leq} \frac{\mu^2r}{n}\left(\abs{\sigma^*_{k+1}}+ 5 \twonorm{\Et} + 4 \mu^2 r \alpha n \infnorm{\Et}\right) \nonumber \\
	&\stackrel{}{\leq} \frac{2 \mu^2r}{n}\left(\abs{\sigma^*_{k+1}} + \left(\frac{1}{2}\right)^{t} \abs{\sigma^*_k}\right),
\end{align*}
where we used Lemma~\ref{lem:spec1} and the assumption on $\infnorm{\Et}$.
\end{proof}

We used the following technical lemma in the proof of Lemma~\ref{lem:dec1}.
\begin{lemma}\label{lem:spec2_noisy}
Assume the notation of Lemma~\ref{lem:dec1}.
Suppose further that
\begin{enumerate}
  \item	$\infnorm{\Et} \leq \frac{8\mu^2 r}{n}\left(|\sigma_{k+1}^*|+\left(\frac{1}{2}\right)^{t-1}|\sigma_k^*|
  \right)$, and
  \item	$\supp{\Et}\subseteq \supp{\So}$.
\end{enumerate}
Then we have:
\begin{equation*}
\twonorm{\Et} \leq 17 \mu^2 r \alpha |\sigma_k^*|, \quad
\abs{\lambda_k} \geq \abs{\sigma_k^*}(1- 17\mu^2 r \alpha), \text{\ \ and\ \  }
\abs{\lambda_{k+1}}\leq |\sigma_{k+1}^*|+\twonorm{\Et}.
\end{equation*}
\end{lemma}
\begin{proof}
Using Lemmas~\ref{lem:spec1} and~\ref{lem:weyl-perturbation}, we have: 
\begin{align*}
|\lambda_i-\sigma_i^*|\leq \|\Et\|_2 \leq \alpha n \infnorm{\Et}.
\end{align*}
The result follows by using the bound on  $\infnorm{\Et}$.
\end{proof}

The following lemma bounds the support of $\Etn$ and $\infnorm{\Etn}$, using an assumption on $\infnorm{\Ltn-\Lo}$.
\begin{lemma}\label{lem:supp-bound}
Assume the notation of Lemma~\ref{lem:dec1}.
Suppose $$\infnorm{\Ltn-\Lo} \leq \frac{2\mu^2 r}{n}\left(|\sigma_{k+1}^*|+\left(\frac{1}{2}\right)^{t}|\sigma_k^*| \right).$$
Then, we have:
\begin{enumerate}
  \item	$\supp{\Etn} \subseteq \supp{\So}$.
  \item	$\infnorm{\Etn} \leq \frac{7\mu^2 r}{n}\left(|\sigma_{k+1}^*|+\left(\frac{1}{2}\right)^{t}|\sigma_k^*| \right)$, and
\end{enumerate}
\end{lemma}
\begin{proof}
We first prove the first conclusion. Recall that, $$\Stn=H_{\zeta}(\M-\Ltn)=H_{\zeta}(\Lo-\Ltn+\So),$$
where $\zeta=\frac{4\mu^2 r}{n}\left(|\lambda_{k+1}|+\left(\frac{1}{2}\right)^{t}|\lambda_k|\right) $ is as defined
in Algorithm~\ref{algo:alg1} and $\lambda_1,\cdots,\lambda_n$ are the eigenvalues of $\M-\St$ such that
$\abs{\lambda_1} \geq \cdots \geq \abs{\lambda_n}$. 

If $\So_{ij}=0$ then $\Etn_{ij}=\ind{\abs{\Lo_{ij}-\Ltn_{ij}}>\zeta}\cdot (\Lo_{ij}-\Ltn_{ij})$. The first part of
the lemma now follows by using the assumption that $\infnorm{\Ltn-\Lo} \leq \frac{2\mu^2 r}{n}\left(|\sigma_{k+1}^*|
+\left(\frac{1}{2}\right)^{t}|\sigma_k^*| \right)\stackrel{(\zeta_1)}{\leq}\frac{4\mu^2 r}{n}\left(|\lambda_{k+1}^*|
+\left(\frac{1}{2}\right)^{t}|\lambda_k^*| \right)=\zeta$, where $(\zeta_1)$ follows from Lemma~\ref{lem:threshold_est}. 

We now prove the second conclusion. We consider the following two cases: 
\begin{enumerate}
  \item  $\abs{\M_{ij}-\Ltn_{ij}} > \zeta$:
	Here, $\Stn_{ij}=\So_{ij}+\Lo_{ij}-\Ltn_{ij}$. Hence, $|\Stn_{ij}-\So_{ij}|\leq |\Lo_{ij}-\Ltn_{ij}|\leq \frac{2\mu^2 r}{n}\left(|\sigma_{k+1}^*|+\left(\frac{1}{2}\right)^{t}|\sigma_k^*| \right)$. 
  \item $\abs{\M_{ij}-\Ltn_{ij}} \leq  \zeta $:
	In this case, $\Stn_{ij}=0$ and $\abs{\So_{ij}+\Lo_{ij}-\Ltn_{ij}} \leq \zeta$.
	So we have, $\abs{\Etn_{ij}}=\abs{\So_{ij}} \leq \zeta+\abs{\Lo_{ij}-\Ltn_{ij}} \leq \frac{7\mu^2 r}{n}\left(|\sigma_{k+1}^*|+\left(\frac{1}{2}\right)^{t}|\sigma_k^*| \right)$.
	The last inequality above follows from Lemma~\ref{lem:threshold_est}.
\end{enumerate}
This proves the lemma.
\end{proof}

We are now ready to prove Lemma~\ref{lem:dec}. In fact, we prove the following stronger version.
\begin{proof}[Proof of Lemma~\ref{lem:dec}]
Recall that in the $k^{\textrm{th}}$ stage, the update $\Ltn$ is given by: $\Ltn=P_k(\M-\St)$ and $\Stn$ is given by: $\Stn=H_\zeta(\M-\Ltn)$. 
Also, recall that $\Et \defas \So-\St$ and $\Etn \defas \So-\Stn$.

We prove the lemma by induction on both $k$ and $t$.
For the base case ($k=1$ and $t=-1$), we first note that the first inequality on $\infnorm{\Lt[0]-\Lo}$ is trivially satisfied.
Due to the thresholding step (step $3$ in Algorithm~\ref{algo:sap}) and the incoherence assumption on $\Lo$, we have: 
\begin{align*}
  \infnorm{\E^{(0)}} &\leq \frac{8\mu^2 r}{n}\left(\sigma_2^* + 2 \sigma_1^*\right), \mbox{ and }\\
  \supp{\Et[0]} &\subseteq \supp{\So}.
\end{align*}
So the base case of induction is satisfied.

We first do the inductive step over $t$ (for a fixed $k$).
By inductive hypothesis we assume that:
a) $\infnorm{\Et} \leq \frac{8\mu^2 r}{n}\left(|\sigma_{k+1}^*|+\left(\frac{1}{2}\right)^{t-1}|\sigma_k^*|
\right)$,
b) $\supp{\Et}\subseteq \supp{\So}$. Then by Lemma~\ref{lem:dec1}, we have:
\begin{align*}
\infnorm{\Ltn-\Lo} \leq \frac{2\mu^2 r}{n}\left(|\sigma_{k+1}^*|+\left(\frac{1}{2}\right)^{t+1}|\sigma_k^*|
\right).
\end{align*}
Lemma~\ref{lem:supp-bound} now tells us that
\begin{enumerate}
  \item	$\infnorm{\Etn} \leq  \frac{8\mu^2 r}{n}\left(|\sigma_{k+1}^*|+\left(\frac{1}{2}\right)^{t}|\sigma_k^*|
  \right)$, and
  \item	$\supp{\Etn} \subseteq \supp{\So}$.
\end{enumerate}
This finishes the induction over $t$. Note that we show a stronger bound than necessary on $\infnorm{\Etn}$.

We now do the induction over $k$. Suppose the hypothesis holds for stage $k$.
Let $T$ denote the number of iterations in each stage. We first obtain a lower bound on $T$.
Since
\begin{align*}
\twonorm{\M-\S^{(0)}} \geq \twonorm{\Lo } - \twonorm{\E^{(0)}} \geq \abs{\sigma_1^*} - \alpha n \infnorm{\E^{(0)}}
\geq \frac{3}{4} \abs{\sigma_1^*},
\end{align*}
we see that $T \geq 10 \log \left(3 \mu^2 r \abs{\sigma_1^*}/\epsilon\right)$.
So, at the end of stage $k$, we have:
\begin{enumerate}
  \item	$\infnorm{\E^{(T)}} \leq \frac{7\mu^2 r}{n}\left(|\sigma_{k+1}^*|+\left(\frac{1}{2}\right)^{T}|\sigma_k^*|
   \right)
	\leq \frac{7\mu^2 r \abs{\sigma_{k+1}^*}}{n} + \frac{\epsilon}{10 n}$, and
  \item	$\supp{\Et[T]} \subseteq \supp{\So}$.
\end{enumerate}
Lemmas~\ref{lem:spec1} and~\ref{lem:weyl-perturbation} tell us that
$\abs{\sigma_{k+1}\left(\M-\St[T]\right) - \abs{\sigma_{k+1}^*}} \leq \twonorm{\E^{(T)}} \leq \alpha\left( 7\mu^2 r \abs{\sigma_{k+1}^*} + \epsilon \right)$.
We will now consider two cases:
\begin{enumerate}
  \item	\textbf{Algorithm~\ref{algo:alg1} terminates:} This means that $\beta \sigma_{k+1}\left(\M-\St[T]\right) <
\frac{\epsilon}{2 n}$ which then implies that $\abs{\sigma_{k+1}^*} < \frac{\epsilon}{6 \mu^2 r}$. So we have:
\begin{align*}
\infnorm{\widehat{L} - \Lo} = \infnorm{\Lt[T] - \Lo}%
	\leq \frac{2 \mu^2 r}{n} \left(|\sigma_{k+1}^*|+\left(\frac{1}{2}\right)^{T}|\sigma_k^*|
	\right) \leq \frac{\epsilon}{5n}.
\end{align*}
This proves the statement about $\widehat{L}$. A similar argument proves the claim on $\infnorm{\widehat{S}-\So}$.
The claim on $\supp{\widehat{S}}$ follows since $\supp{\Et[T]} \subseteq \supp{\So}$.
  \item	\textbf{Algorithm~\ref{algo:alg1} continues to stage $\bm{\left(k+1\right)}$:} This means that
$\beta \sigma_{k+1}\left(\Lt[T]\right) \geq \frac{\epsilon}{2 n}$ which then implies that
$\abs{\sigma_{k+1}^*} > \frac{\epsilon}{8 \mu^2 r}$. So we have:
\begin{align*}
\infnorm{\E^{(T)}} &\leq \frac{8\mu^2 r}{n}\left(|\sigma_{k+1}^*|+\left(\frac{1}{2}\right)^{T}|\sigma_k^*|
\right) \\
	&\leq \frac{8\mu^2 r}{n}\left(|\sigma_{k+1}^*|+ \frac{\epsilon}{10\mu^2 r n}
\right) \\
	&\leq \frac{8\mu^2 r}{n}\left(|\sigma_{k+1}^*|+ \frac{8 \abs{\sigma_{k+1}^*}}{10 n}
	 \right) \\
	&\leq \frac{8\mu^2 r}{n}\left(|\sigma_{k+2}^*|+ 2 \abs{\sigma_{k+1}^*}
	\right).
\end{align*}
Similarly for $\infnorm{\Lt[T]-\Lo}$.
\end{enumerate}
This finishes the proof.
\end{proof}

\begin{proof}[Proof of Theorem~\ref{thm:main}]
Using Lemma~\ref{lem:dec}, it suffices to show that the general case can be reduced to the case of symmetric
matrices. We will now outline this reduction.

Recall that we are given an $m \times n$ matrix $\M=\Lo + \So $ where $\Lo$ is the true low-rank matrix
and $\So$ the sparse error matrix. Wlog, let $m \leq n$ and suppose $ \beta m \leq n < (\beta+1) m$, for
some $\beta \geq 1$. We then consider the symmetric matrices
\begin{align}
&  \widetilde{\M} =
\underbrace{
\left[ 	\begin{array}{c}
		\begin{array}{c}
			0 \\ \vdots \\ 0
		\end{array}
		\cdots 
		\begin{array}{c}
			0 \\ \vdots \\ 0
		\end{array}	\\
		\trans{M} \cdots \trans{M}
	\end{array}
\right.}_{\beta \mbox{ times}}
\left.	\begin{array}{c}
		\begin{array}{c}
			M \\ \vdots \\ M
		\end{array} \\
		0
      	\end{array}
\right],
  \widetilde{\L} =
\underbrace{
\left[ 	\begin{array}{c}
		\begin{array}{c}
			0 \\ \vdots \\ 0
		\end{array}
		\cdots 
		\begin{array}{c}
			0 \\ \vdots \\ 0
		\end{array}	\\
		\trans{\left(\Lo\right)} \cdots \trans{\left(\Lo\right)}
	\end{array}
\right.}_{\beta \mbox{ times}}
\left.	\begin{array}{c}
		\begin{array}{c}
			\Lo \\ \vdots \\ \Lo
		\end{array} \\
		0
      	\end{array}
\right],
\label{eqn:symmetrization_app}
\end{align}
and $\wS = \wM - \wL$. A simple calculation shows that $\wL$ is incoherent with parameter $\sqrt{3}\mu$
and $\wS$ satisfies the sparsity condition (S1) with parameter $\frac{\alpha}{\sqrt{2}}$.
Moreover the iterates of \ncralgo~ with input $\wM$ have similar expressions as in \eqref{eqn:symmetrization_app} in terms of the corresponding
iterates with input $\M$. This means that it suffices to obtain the same guarantees for Algorithm~\ref{algo:alg1}
for the symmetric case. Lemma~\ref{lem:dec} does precisely this, proving the theorem.
\end{proof}




\section{Proof of Theorem~\ref{thm:noise}}
\label{app:noise}
In this section, we prove Theorem~\ref{thm:noise}. The roadmap of the proofs in this section is essentially the
same as that in Appendix~\ref{app:proofs}.

In what follows, we prove a number of lemmas concerning the structure of $\Lt$ and $\Et \defas \So-\St$.
The first lemma is a generalization of Lemma~\ref{lem:threshold_est} and shows that the threshold in \eqref{eqn:threshold} is close to
that with $\Mt$ replaced by $\Lo$.

\begin{lemma}\label{lem:threshold_est_noise}
Let $\Lo, \So, \No$ be symmetric and satisfy the assumptions of Theorem~\ref{thm:noise} and let $\St$ be the
$t^{\textrm{th}}$ iterate of the $k^{\textrm{th}}$ stage of Algorithm~\ref{algo:alg1}. Let $\sigma_1^*, \dots,
\sigma_r^*$ be the eigenvalues of $\Lo$, such that $|\sigma_1^*|\geq \dots \geq |\sigma_r^*|$ and $\lambda_1,
\cdots,\lambda_n$ be the eigenvalues of $\M-\St$ such that $\abs{\lambda_1} \geq \cdots \geq \abs{\lambda_n}$.
Recall that $\Et \defas \So - \St$. Suppose further that
\begin{enumerate}
  \item	$\infnorm{\Et} \leq \frac{8\mu^2 r}{n}\left(|\sigma_{k+1}^*|+\left(\frac{1}{2}\right)^{t-1}|\sigma_k^*|
  + 7 \twonorm{\No} + \frac{8n}{\sqrt{r}} \infnorm{\No}\right)$, and
  \item	$\supp{\St} \subseteq \supp{\So}$.
\end{enumerate}
Then,
\begin{align}
  \label{eq:supp_1_noise}
  \frac{7}{8}\left(|\sigma_{k+1}^*|+\left(\frac{1}{2}\right)^{t}|\sigma_k^*|\right)
  \leq \left(|\lambda_{k+1}|+\left(\frac{1}{2}\right)^{t}|\lambda_k|\right)
  \leq \frac{9}{8} \left(|\sigma_{k+1}^*|+\left(\frac{1}{2}\right)^{t}|\sigma_k^*|\right).
\end{align}
\end{lemma}
\begin{proof}
Note that $\M-\St=\Lo+\No+\Et$. Now, using Lemmas~\ref{lem:weyl-perturbation} and~\ref{lem:spec1}, we have:
\begin{align*}
\abs{\lambda_{k+1}-\sigma_{k+1}^*} \leq \twonorm{\Et} \leq \alpha n \infnorm{\Et} \leq 8\mu^2 r \alpha \gamma_t,
\end{align*}
where $\gamma_t=\left(|\sigma_{k+1}^*|+\left(\frac{1}{2}\right)^{t-1}|\sigma_k^*| + 7 \twonorm{\No} + \frac{8n}{\sqrt{r}} \infnorm{\No}\right)$.
That is,
$\abs{|\lambda_{k+1}|-|\sigma_{k+1}^*|}\leq  8\mu^2 r \alpha \gamma_t$.
Similarly, $\abs{|\lambda_{k}|-|\sigma_{k}^*|} \leq  8\mu^2 r \alpha \gamma_t$.
So we have:
\begin{align*}
\abs{ \left(|\lambda_{k+1}|+\left(\frac{1}{2}\right)^{t}|\lambda_k|\right) - \left(|\sigma_{k+1}^*|+\left(\frac{1}{2}\right)^{t}|\sigma_k^*|\right)} &\leq 8\mu^2 r \alpha \gamma_t \left(1+\left(\frac{1}{2}\right)^t\right) \\
&\leq 16 \mu^2 r \alpha \gamma_t \\
&\leq \frac{1}{8}\left(|\sigma_{k+1}^*|+\left(\frac{1}{2}\right)^{t}|\sigma_k^*|\right),
\end{align*}
where the last inequality follows from the bound $\alpha\leq \frac{1}{512\mu^2 r}$
and the assumption on $\infnorm{\No}$.
\end{proof}

The following lemma shows that under the same assumptions as in Lemma~\ref{lem:threshold_est},
we can obtain a bound on the $\ell_{\infty}$ norm of $\Ltn-\Lo$.
This is the most crucial step in our analysis since we bound $\ell_{\infty}$ norm of errors which are quite hard
to obtain.
\begin{lemma}\label{lem:dec1_noise}
Assume the notation of Lemma~\ref{lem:threshold_est}. Also, let $\Lt, \St$ be the $t^{\textrm{th}}$ iterates of
$k^{\textrm{th}}$ stage of Algorithm~\ref{algo:alg1} and $\Ltn, \Stn$ be the $(t+1)^{\textrm{th}}$ iterates of the same
stage. Also, recall that $\Et \defas \So-\St$ and $\Etn \defas \So-\Stn$.
Suppose further that
\begin{enumerate}
  \item	$\infnorm{\Et} \leq \frac{8\mu^2 r}{n}\left(|\sigma_{k+1}^*|+\left(\frac{1}{2}\right)^{t-1}|\sigma_k^*| + 7 \twonorm{\No} + \frac{8n}{\sqrt{r}} \infnorm{\No} \right)$, and
  \item	$\supp{\Et}\subseteq \supp{\So}$.
\end{enumerate}
Then, we have:
  $$\infnorm{\Ltn-\Lo} \leq \frac{2\mu^2 r}{n}\left(|\sigma_{k+1}^*|+\left(\frac{1}{2}\right)^{t}|\sigma_k^*| + 7 \twonorm{\No} + \frac{8n}{\sqrt{r}} \infnorm{\No} \right).$$
\end{lemma}

\begin{proof}
  Let $\Ltn=P_k(\M-\St)=U\Lambda \trans{U}$ be the eigenvalue decomposition of $\Ltn$. Also, recall that
  $\M-\St=\Lo+\No+\Et$. Then, for every eigenvector $\u_i$ of $\Ltn$, we have
  \begin{align*}
    \left(\Lo+\No+\Et\right)\u_i&=\lambda_i \u_i, \\
    \left(I-\frac{\Et}{\lambda_i}\right)\u_i&=\frac{1}{\lambda_i}\left(\Lo+\No\right)\u_i, \\
\u_i&=\left(I-\frac{\Et}{\lambda_i}\right)^{-1}\frac{\left(\Lo+\No\right)\u_i}{\lambda_i} \\
&=\left(I+\frac{\Et}{\lambda_i}+\left(\frac{\Et}{\lambda_i}\right)^2+\dots\right)\frac{\left(\Lo+\No\right)\u_i}{\lambda_i}.
  \end{align*}
Note that we used Lemmas~\ref{lem:weyl-perturbation} and~\ref{lem:spec1} to guarantee the existence of
$\inv{\left(I - \frac{\Et}{\lambda_i}\right)}$.
Hence, 
\begin{align*}
  U\Lambda \trans{U} - \Lo &= (\left(\Lo+\No\right) U \Lambda^{-1} \trans{U} \left(\Lo+\No\right) -\Lo)\\
  &\qquad + \sum_{p+q\geq 1} (\St)^p \left(\Lo+\No\right) U \Lambda^{-(p+q+1)} \trans{U} \left(\Lo+\No\right) (\St)^q. 
\end{align*}
By triangle inequality, we have
\begin{align}
    \infnorm{U\Lambda \trans{U} - \Lo} &\leq \infnorm{ \left(\Lo+\No\right) U \Lambda^{-1} \trans{U} \left(\Lo+\No\right) -\Lo} \nonumber \\
    &\quad + \sum_{p+q\geq 1} \infnorm{(\St)^p \left(\Lo+\No\right) U \Lambda^{-(p+q+1)} \trans{U} \left(\Lo+\No\right) (\St)^q}. 
  \label{eq:dec1_n_noise}
\end{align}
We now bound the two terms on the right hand side above.

For the first term, we again use triangle inequality to obtain
\begin{align}
    \infnorm{\left(\Lo+\No\right) U \Lambda^{-1} \trans{U} \left(\Lo+\No\right) -\Lo}
    &\leq \infnorm{\Lo U \Lambda^{-1} \trans{U} \Lo - \Lo} + \infnorm{\No U \Lambda^{-1} \trans{U} \Lo} \nonumber\\
	&\qquad + \infnorm{\Lo U \Lambda^{-1} \trans{U} \No} + \infnorm{\No U \Lambda^{-1} \trans{U} \No}. \label{eqn:firstterm-exp_n_noise}
\end{align}
We note that, 
\begin{align}
  &\infnorm{\Lo U \Lambda^{-1} \trans{U} \Lo -\Lo} \nonumber \\
  &=\max_{ij}\trans{\e_i}\left(\Uo\Sigmao\trans{(\Uo)}U \Lambda^{-1} \trans{U}\Uo\Sigmao\trans{(\Uo)}-\Uo\Sigmao\trans{(\Uo)}\right)\e_j\nonumber\\
  &=\max_{ij}\trans{\e_i}\Uo\left(\Sigmao\trans{(\Uo)}U \Lambda^{-1} \trans{U}\Uo\Sigmao - \Sigmao \right) \trans{(\Uo)}\e_j\nonumber\\
&\leq \max_{ij}\|\trans{\e_i}\Uo\|\cdot \|\trans{\e_j}\Uo\|\cdot \|\Uo\Sigmao\trans{(\Uo)}U \Lambda^{-1} \trans{U}\Uo\Sigmao\trans{(\Uo)}-\Uo\Sigmao\trans{(\Uo)}\|_2\nonumber\\
&\leq \frac{\mu^2r}{n}\|\Lo U \Lambda^{-1} \trans{U}\Lo-\Lo\|_2,\label{eq:dec2_n_noise}
\end{align}
where we denote $\Uo\Sigmao\trans{(\Uo)}$ to be the SVD of $\Lo$.
Let $\Lo+\No+\Et = U\Lambda \trans{U} + \widetilde{U} \widetilde{\Lambda} \trans{\widetilde{U}}$ be the eigenvalue decomposition
of $\Lo+\No+\Et$.
Note that $\trans{\widetilde{U}}U=0$. Recall that, $U\Lambda \trans{U}=P_k(\Mt)=\Lt$. Also note that,
\begin{align}
  &\Lo U \Lambda^{-1} \trans{U}\Lo-\Lo \nonumber \\
  & = (U\Lambda \trans{U} + \widetilde{U} \widetilde{\Lambda} \trans{\widetilde{U}}-\No-\Et)U\Lambda^{-1} \trans{U}
(U\Lambda \trans{U} + \widetilde{U} \widetilde{\Lambda} \trans{\widetilde{U}} -\No-\Et) - \Lo,\nonumber\\
&=(U\trans{U}-\left(\No+\Et\right) U\Lambda^{-1} \trans{U}) (U\Lambda \trans{U} + \widetilde{U} \widetilde{\Lambda} \trans{\widetilde{U}}-\No-\Et) - \Lo,\nonumber\\
&=U\Lambda \trans{U}-U\trans{U}\left(\No+\Et\right) - \left(\No+\Et\right) U\trans{U}  \nonumber\\
&-\left(\No+\Et\right) U\Lambda^{-1}\trans{U} \trans{\left(\No+\Et\right)}- U\Lambda \trans{U} -\widetilde{U} \widetilde{\Lambda} \trans{\widetilde{U}}+\No+\Et.\label{eq:dec6_n_noise}
\end{align}
Hence, using Lemma~\ref{lem:spec2_noise}, we have: 
\begin{align}
  \twonorm{\Lo U \Lambda^{-1} \trans{U}\Lo - \Lo} &\leq 3\twonorm{\No+\Et} + \frac{\twonorm{\No+\Et}^2}{\abs{\lambda_k}}+\abs{\lambda_{k+1}} \nonumber\\
&\leq \abs{\sigma^*_{k+1}}+4\twonorm{\No+\Et} +\frac{\twonorm{\No+\Et}^2}{(1-17\mu^2 r \alpha)\abs{\sigma^*_k}}.
  \label{eq:dec3_n_noise}
\end{align}

Using \eqref{eq:dec2_n_noise}, \eqref{eq:dec3_n_noise}, and Lemma~\ref{lem:spec2_noise}:
\begin{equation}
  \infnorm{\Lo U \Lambda^{-1} \trans{U} \Lo -\Lo} \leq \frac{\mu^2r}{n}\left(\abs{\sigma^*_{k+1}} + 7 \twonorm{\No}
  + 5 \twonorm{\Et}\right)
  \label{eqn:firstterm-final_n_noise}
\end{equation}

Coming to the second term of \eqref{eqn:firstterm-exp_n_noise}, we have:
\begin{align}
 &\infnorm{\No U \Lambda^{-1} \trans{U} \Lo} \nonumber \\
 &= \max_{i,j} \trans{\e_i} \No U \Lambda^{-1} \trans{U} \Lo \e_j \nonumber\\
	&\leq \max_i\twonorm{\trans{\e_i} \No U} \twonorm{\Lambda^{-1}\trans{U} \Uo \Sigmao} \twonorm{\trans{(\Uo)} \e_j} \nonumber\\
	&\leq \sqrt{n} \infnorm{\No} \twonorm{\Lambda^{-1}\trans{U} \Uo \Sigmao} \frac{\mu \sqrt{r}}{\sqrt{n}}
	= \mu \sqrt{r} \infnorm{\No} \twonorm{U\Lambda^{-1}\trans{U} \Uo \Sigmao \trans{(\Uo)}}. \label{eqn:secondterm-1_n_noise}
\end{align}
Using an expansion along the lines of \eqref{eq:dec6_n_noise}, we see that
\begin{align*}
  \twonorm{U\Lambda^{-1}\trans{U} \Uo \Sigmao \trans{(\Uo)}}
	\leq 1 + \frac{\twonorm{\No + \Et}}{\abs{\lambda_k}}
	&\leq 1 + \frac{\twonorm{\No} + \twonorm{\Et}}{(1-17\mu^2r\cdot \alpha)\abs{\sigma^*_k}} \\
	&\leq 2 + \frac{\twonorm{\Et}}{(1-17\mu^2 r \alpha)\abs{\sigma^*_k}}.
\end{align*}
Plugging this in \eqref{eqn:secondterm-1_n_noise} gives us
\begin{align}
 \infnorm{\No U \Lambda^{-1} \trans{U} \Lo} &\leq 3 \mu \sqrt{r} \infnorm{\No}. \label{eqn:secondterm-final_n_noise}
\end{align}
A similar argument as in \eqref{eqn:secondterm-1_n_noise} gives us the following bound on the last term in \eqref{eqn:firstterm-exp_n_noise}:
\begin{align}
 \infnorm{\No U \Lambda^{-1} \trans{U} \No} &\leq \frac{n \infnorm{\No}^2}{\abs{\lambda_k}} \leq \infnorm{\No}. \label{eqn:thirdterm-final_n_noise}
\end{align}

Plugging \eqref{eqn:firstterm-final_n_noise},~\eqref{eqn:secondterm-final_n_noise} and~\eqref{eqn:thirdterm-final_n_noise}, we obtain:
\begin{align}
    &\infnorm{\left(\Lo+\No\right) U \Lambda^{-1} \trans{U} \left(\Lo+\No\right) -\Lo} \nonumber \\
	&\qquad \qquad \qquad \leq \frac{\mu^2r}{n}\left(\abs{\sigma^*_{k+1}}+ 7 \twonorm{\No} + 7 \twonorm{\Et} + \frac{7 n}{\sqrt{r}} \infnorm{\No}\right).
  \label{eq:dec8_n_noise}
\end{align}

Next, we analyze $\sum_{p+q\geq 1} \infnorm{(\Et)^p (\Lo+\No) U \Lambda^{-(p+q+1)} \trans{U} (\Lo+\No) (\Et)^q}$.
This can again be bounded by four quantities:
\begin{align}
&\infnorm{(\Et)^p (\Lo+\No) U \Lambda^{-(p+q+1)} \trans{U} (\Lo+\No) (\Et)^q} \nonumber\\
	&\qquad\leq \infnorm{(\Et)^p \Lo U \Lambda^{-(p+q+1)} \trans{U} \Lo (\Et)^q}
	+ \infnorm{(\Et)^p \No U \Lambda^{-(p+q+1)} \trans{U} \Lo (\Et)^q} \\
	&\qquad\quad+ \infnorm{(\Et)^p \Lo U \Lambda^{-(p+q+1)} \trans{U} \No (\Et)^q}
	+ \infnorm{(\Et)^p \No U \Lambda^{-(p+q+1)} \trans{U} \No (\Et)^q}. \label{eqn:powerseries-triineq_n_noise}
\end{align}

We bound the first term above:
\begin{align}
  &\infnorm{(\Et)^p \Lo U \Lambda^{-(p+q+1)} \trans{U} \Lo (\Et)^q} \nonumber \\
  &\qquad \qquad = \max_{ij} \trans{\e_i} \left((\Et)^p \Lo U \Lambda^{-(p+q+1)} \trans{U} \Lo (\Et)^q\right)\e_j,\nonumber\\
&\qquad \qquad \leq  \max_{ij}\twonorm{\trans{\e_i}(\Et)^p\Uo} \twonorm{\trans{\e_j}(\Et)^q\Uo} \twonorm{\Lo U \Lambda^{-(p+q+1)} \trans{U} \Lo},\nonumber\\
&\qquad \qquad \stackrel{(\zeta_1)}{\leq} \frac{\mu^2r}{n} \left(\alpha n \infnorm{\Et}\right)^p \left(\alpha n  \infnorm{\Et}\right)^q 
\twonorm{\Lo U \Lambda^{-(p+q+1)} \trans{U} \Lo},\label{eq:dec9_n_noise}
\end{align}
where $(\zeta_1)$ follows from Lemma~\ref{lem:inf} and the incoherence of $\Lo$.
Now, similar to \eqref{eq:dec6_n_noise}, we have: 
\begin{align}
&  \twonorm{\Lo U \Lambda^{-(p+q+1)} \trans{U} \Lo} \nonumber \\
&\qquad = \left\|U\Lambda^{-(p+q-1)}\trans{U}-\left(\No+\Et\right) U\Lambda^{-(p+q)}\trans{U}-
	U\Lambda^{-(p+q)}\trans{U}\left(\No+\Et\right)\right. \nonumber \\
&\qquad \qquad \left. +\left(\No+\Et\right) U\Lambda^{-(p+q+1)}\trans{U} \left(\No+\Et\right)\right\|_2,\nonumber\\
&\qquad \leq \|\Lambda^{-(p+q-1)}\|_2+2\|\No+\Et\|_2 \|\Lambda^{-(p+q)}\|_2+\|\No+\Et\|_2^2 \|\Lambda^{-(p+q+1)}\|_2,\nonumber\\
&\qquad \leq \abs{\lambda_k}^{-(p+q-1)}\left(1+2\frac{\|\No+\Et\|_2}{\abs{\lambda_k}}+\frac{\|\No+\Et\|_2^2}{\lambda_k^2}\right),\nonumber\\
&\qquad = \abs{\lambda_k}^{-(p+q-1)}\left(1+\frac{\|\No+\Et\|_2}{\abs{\lambda_k}}\right)^2,\nonumber\\
&\qquad \leq \abs{\lambda_k}^{-(p+q-1)}\left(1+\frac{\twonorm{\No}+\twonorm{\Et}}{\abs{\lambda_k}}\right)^2,\nonumber\\
&\qquad \stackrel{(\zeta_1)}{\leq}  \abs{\lambda_k}^{-(p+q-1)}\left(1+\frac{\twonorm{\No}+17 \mu^2 r \alpha \abs{\sigma_k^*}}{(1-17 \mu^2 r \alpha)\abs{\sigma_k^*}}\right)^2
{\leq}  2 \abs{\lambda_k}^{-(p+q-1)},\label{eq:dec10_n_noise}
\end{align}
where $(\zeta_1)$ follows from Lemma~\ref{lem:spec2_noise} and the bound on $\infnorm{\No}$.

Using \eqref{eq:dec9_n_noise}, \eqref{eq:dec10_n_noise}, we have: 
\begin{align}
  \infnorm{(\Et)^p \Lo U \Lambda^{-(p+q+1)} \trans{U} \Lo (\Et)^q}
  \leq 2\alpha \mu^2 r \infnorm{\Et} \left(\frac{\alpha n \infnorm{\Et}}{\abs{\lambda_k}}\right)^{p+q-1}.
\label{eqn:ps-firstterm-final_n_noise}
\end{align}

Coming to the second term of \eqref{eqn:powerseries-triineq_n_noise}, we have
\begin{align}
\hspace*{50pt}&\hspace*{-60pt}  \infnorm{(\Et)^p \No U \Lambda^{-(p+q+1)} \trans{U} \Lo (\Et)^q}\nonumber\\
&= \max_{i,j} \trans{\e_i} \left((\Et)^p \No U \Lambda^{-(p+q+1)} \trans{U} \Lo (\Et)^q\right)\e_j,\nonumber\\
&\leq  \max_{ij}\twonorm{\trans{\e_i}(\Et)^p\No U} \twonorm{\trans{\e_j}(\Et)^q\Uo} \twonorm{\Lambda^{-(p+q+1)} \trans{U} \Lo}\nonumber\\
&\stackrel{(\zeta_1)}{\leq} \frac{\mu\sqrt{r}}{\sqrt{n}} \infnorm{\No U} \left(\alpha n \infnorm{\Et} \right)^p \left(\alpha n  \infnorm{\Et} \right)^q \twonorm{U \Lambda^{-(p+q+1)} \trans{U} \Lo} \nonumber \\
&\stackrel{}{\leq} \mu\sqrt{r} \infnorm{\No} \left(\alpha n \infnorm{\Et}\right)^{p+q} \twonorm{U \Lambda^{-(p+q+1)} \trans{U} \Lo},\label{eqn:ps-secondterm-1_n_noise}
\end{align}
where $(\zeta_1)$ follows from Lemma~\ref{lem:inf} and incoherence of $\Uo$.
Proceeding along the lines of \eqref{eq:dec10_n_noise}, we obtain:
\begin{align*}
  \twonorm{U \Lambda^{-(p+q+1)} \trans{U} \Lo}
&{\leq}  \abs{\lambda_k}^{-(p+q)}\left(1+\frac{\twonorm{\No}+\twonorm{\Et}}{\abs{\lambda_k}}\right)
{\leq}  2 \abs{\lambda_k}^{-(p+q)}.
\end{align*}
Plugging the above in \eqref{eqn:ps-secondterm-1_n_noise} gives us
\begin{align}
&\infnorm{(\Et)^p \No U \Lambda^{-(p+q+1)} \trans{U} \Lo (\Et)^q}
\leq 2 \mu \sqrt{r} \left(\frac{\alpha n \infnorm{\Et}}{\abs{\lambda_k}}\right)^{p+q} \infnorm{\No}.\label{eqn:ps-secondterm-final_n_noise}
\end{align}

A similar argument as in \eqref{eqn:ps-secondterm-1_n_noise} gives us
\begin{align*}
\infnorm{(\Et)^p \No U \Lambda^{-(p+q+1)} \trans{U} \No (\Et)^q}
&\leq \frac{n \infnorm{\No}}{\abs{\lambda_k}} \left(\frac{\alpha n \infnorm{\Et}}{\abs{\lambda_k}}\right)^{p+q}  \infnorm{\No}.
\end{align*}

Plugging the above inequality along with \eqref{eqn:ps-firstterm-final_n_noise} and \eqref{eqn:ps-secondterm-final_n_noise}
into \eqref{eqn:powerseries-triineq_n_noise} gives us:
\begin{align*}
&\infnorm{(\Et)^p (\Lo+\No) U \Lambda^{-(p+q+1)} \trans{U} (\Lo+\No) (\Et)^q} \\
	&\qquad \qquad \leq 2 \mu^2 r \left(\alpha \infnorm{\Et} + \frac{\infnorm{\No}}{\sqrt{r}}\right) 
\left(\frac{\alpha n \infnorm{\Et}}{\abs{\lambda_k}}\right)^{p+q-1}.
\end{align*}
Using the above bound, and the assumption on $\infnorm{\Et}$:
\begin{align*}
\infnorm{\Et} \leq \frac{8\mu^2r}{n}\left(\abs{\sigma_{k+1}^*}+
\left(\frac{1}{2}\right)^{t-1}\abs{\sigma_k^*}
+ 7 \twonorm{\No} + \frac{8n}{\sqrt{r}} \infnorm{\No} \right) \leq \frac{17 \mu^2 r}{n}\abs{\sigma_k^*},
\end{align*}
we have:
\begin{align}
&\sum_{p+q\geq 1}\infnorm{(\Et)^p \left(\Lo+\No\right) U \Lambda^{-(p+q+1)} \trans{U} \left(\Lo+\No\right) (\Et)^q} \\
&\qquad \qquad \leq   2 \mu^2 r \left(\alpha \infnorm{\Et} + \frac{ \infnorm{\No}}{\sqrt{r}}\right) 
	\sum_{p+q\geq 1}\left(\frac{\alpha n \infnorm{\Et}}{\abs{\lambda_k}}\right)^{p+q-1} \nonumber\\
&\qquad \qquad \leq 2 \mu^2 r \left(\alpha \infnorm{\Et} + \frac{\infnorm{\No}}{\sqrt{r}}\right) \left(\frac{1}{1-\frac{17\mu^2 \alpha r}{1-17\mu^2\alpha\cdot r}}\right)^2 \nonumber\\
&\qquad \qquad \leq 2 \mu^2 r \left(\alpha \infnorm{\Et} + \frac{\infnorm{\No}}{\sqrt{r}}\right) \left(\frac{1}{1-34\mu^2 r \alpha}\right)^2 \nonumber\\
&\qquad \qquad \leq 4 \mu^2 r \left( \alpha \infnorm{\Et} + \frac{\infnorm{\No}}{\sqrt{r}}\right). \label{eq:dec12_n_noise}
\end{align}

Combining \eqref{eq:dec1_n_noise}, \eqref{eq:dec8_n_noise}, \eqref{eq:dec12_n_noise}, we have: 
\begin{align*}
  \infnorm{U\Lambda \trans{U} - \Lo}
	&\stackrel{(\zeta_1)}{\leq} \frac{\mu^2r}{n}\left(\abs{\sigma^*_{k+1}}+ 7 \twonorm{\No} + 11 n \alpha \infnorm{\Et} + \frac{11 n}{\sqrt{r}} \infnorm{\No}\right) \nonumber \\
	&\stackrel{(\zeta_2)}{\leq} \frac{2 \mu^2r}{n}\left(\abs{\sigma^*_{k+1}} + \left(\frac{1}{2}\right)^{t} \abs{\sigma^*_k} + 7 \twonorm{\No} + \frac{8 n}{\sqrt{r}} \infnorm{\No}\right),
\end{align*}
where $(\zeta_1)$ follows from Lemma~\ref{lem:spec1}, and $(\zeta_2)$ follows from the assumption on $\infnorm{\Et}$.
\end{proof}

We used the following technical lemma in the proof of Lemma~\ref{lem:dec1_noise}.
\begin{lemma}\label{lem:spec2_noise}
Assume the notation of Lemma~\ref{lem:dec1_noise}.
Suppose further that
\begin{enumerate}
  \item	$\infnorm{\Et} \leq \frac{8\mu^2 r}{n}\left(|\sigma_{k+1}^*|+\left(\frac{1}{2}\right)^{t-1}|\sigma_k^*|
  + 7 \twonorm{\No} + \frac{8n}{\sqrt{r}} \infnorm{\No} \right)$, and
  \item	$\supp{\Et}\subseteq \supp{\So}$.
\end{enumerate}
Then we have:
\begin{equation*}
\twonorm{\Et} \leq 17 \mu^2 r \alpha |\sigma_k^*|, \quad
\abs{\lambda_k} \geq \abs{\sigma_k^*}(1- 17\mu^2 r \alpha), \text{\ \ and\ \  }
\abs{\lambda_{k+1}}\leq |\sigma_{k+1}^*|+\twonorm{\Et}.
\end{equation*}
\end{lemma}
\begin{proof}
Using Lemmas~\ref{lem:spec1} and~\ref{lem:weyl-perturbation}, we have: 
\begin{align*}
|\lambda_i-\sigma_i^*|\leq \twonorm{\Et} \leq \alpha n \infnorm{\Et}.
\end{align*}
Using the bound on $\infnorm{\Et}$ and recalling the assumption that
\begin{align*}
  \infnorm{\No} \leq \frac{\abs{\sigma^*_r}}{100}
\end{align*}
finishes the proof.
\end{proof}

The following lemma bounds the support of $\Etn$ and $\infnorm{\Etn}$, using an assumption on $\infnorm{\Ltn-\Lo}$.
\begin{lemma}\label{lem:supp-bound_noise}
Assume the notation of Lemma~\ref{lem:dec1_noise}.
Suppose
\begin{align*}
\infnorm{\Ltn-\Lo} \leq \frac{2\mu^2 r}{n}\left(|\sigma_{k+1}^*|+\left(\frac{1}{2}\right)^{t}|\sigma_k^*| + 7 \twonorm{\No} + \frac{8n}{\sqrt{r}} \infnorm{\No} \right).
\end{align*}
Then, we have:
\begin{enumerate}
  \item	$\supp{\Etn} \subseteq \supp{\So}$.
  \item	$\infnorm{\Etn} \leq \frac{7\mu^2 r}{n}\left(|\sigma_{k+1}^*|+\left(\frac{1}{2}\right)^{t}|\sigma_k^*| + 7 \twonorm{\No} + \frac{8n}{\sqrt{r}} \infnorm{\No} \right)$, and
\end{enumerate}
\end{lemma}
\begin{proof}
We first prove the first conclusion. Recall that, $$\Stn=H_{\zeta}(\M-\Ltn)=H_{\zeta}(\Lo-\Ltn+\No+\So),$$
where $\zeta=\frac{4\mu^2 r}{n}\left(|\lambda_{k+1}|+\left(\frac{1}{2}\right)^{t}|\lambda_k|\right) $ is as defined
in Algorithm~\ref{algo:alg1} and $\lambda_1,\cdots,\lambda_n$ are the eigenvalues of $\M-\St$ such that
$\abs{\lambda_1} \geq \cdots \geq \abs{\lambda_n}$. 

If $\So_{ij}=0$ then $\Etn_{ij}=\ind{\abs{\Lo_{ij}-\Ltn_{ij}+\No_{ij}}>\zeta} \cdot (\Lo_{ij}-\Ltn_{ij} + \No_{ij})$.
The first part of the lemma now follows by using the assumption that $\infnorm{\Ltn-\Lo} \leq \frac{2\mu^2 r}{n}\left(|\sigma_{k+1}^*|
+\left(\frac{1}{2}\right)^{t}|\sigma_k^*| \right)\stackrel{(\zeta_1)}{\leq}\frac{9\mu^2 r}{4n}\left(|\lambda_{k+1}^*|
+\left(\frac{1}{2}\right)^{t}|\lambda_k^*| \right)=\zeta$, where $(\zeta_1)$ follows from Lemma~\ref{lem:threshold_est},
and the bound on $\infnorm{\No}$.

We now prove the second conclusion. We consider the following two cases: 
\begin{enumerate}
  \item  $\abs{\M_{ij}-\Ltn_{ij}} > \zeta$:
	Here, $\Stn_{ij}=\So_{ij}+\Lo_{ij}-\Ltn_{ij}+\No_{ij}$. Hence, $|\Stn_{ij}-\So_{ij}|\leq |\Lo_{ij}-\Ltn_{ij}|
	+ \abs{\No_{ij}}\leq \frac{2\mu^2 r}{n}\left(|\sigma_{k+1}^*|+\left(\frac{1}{2}\right)^{t}|\sigma_k^*| \right)
	+ \infnorm{\No}$. 
  \item $\abs{\M_{ij}-\Ltn_{ij}} \leq  \zeta $:
	In this case, $\Stn_{ij}=0$ and $\abs{\So_{ij}+\Lo_{ij}-\Ltn_{ij} + \No_{ij}} \leq \zeta$.
	So we have, $\abs{\Etn_{ij}}=\abs{\So_{ij}} \leq \zeta+\abs{\Lo_{ij}-\Ltn_{ij}} + \abs{\No_{ij}}
	\leq \frac{7\mu^2 r}{n}\left(|\sigma_{k+1}^*|+\left(\frac{1}{2}\right)^{t}|\sigma_k^*| \right) + \infnorm{\No}$.
	The last inequality above follows from Lemma~\ref{lem:threshold_est}.
\end{enumerate}
This proves the lemma.
\end{proof}

The following lemma is a generalization of Lemma~\ref{lem:dec}.
\begin{lemma}\label{lem:dec_noisy}
Let $\Lo, \So, \No$ be symmetric and satisfy the assumptions of Theorem~\ref{thm:noise} and let $\Mt$ and $\Lt$ be the
$t^{\textrm{th}}$ iterates of the $k^{\textrm{th}}$ stage of Algorithm~\ref{algo:alg1}. Let $\sigma_1^*, \dots, \sigma_n^*$ be the eigenvalues
of $\Lo$, s.t., $|\sigma_1^*|\geq \dots \geq |\sigma_r^*|$.
Then, the following holds:
\begin{align*}
\infnorm{\Ltn-\Lo} \leq \frac{2\mu^2 r}{n}&\left(\abs{\sigma_{k+1}^*}+\left(\frac{1}{2}\right)^{t}\abs{\sigma_k^*}
+ 7 \twonorm{\No} + \frac{8n}{\sqrt{r}} \infnorm{\No} \right), \\
\infnorm{\Etn}=\infnorm{\So-\Stn} \leq \frac{8\mu^2 r}{n}&\left(\abs{\sigma_{k+1}^*}+\left(\frac{1}{2}\right)^{t-1}\abs{\sigma_k^*}
	+ 7 \twonorm{\No} + \frac{8n}{\sqrt{r}} \infnorm{\No}\right), \mbox{ and } \\
\supp{\Etn} &\subseteq \supp{\So}.
\end{align*}

Moreover, the outputs $\widehat{\L}$ and $\widehat{\S}$ of Algorithm~\ref{algo:alg1} satisfy:
\begin{align*}
\frob{\widehat{\L} - \Lo} &\leq \epsilon + {2 \mu^2r}\left(7 \twonorm{\No} + \frac{8 n}{\sqrt{r}} \infnorm{\No}\right), \\
\infnorm{\widehat{\S} - \So} &\leq \frac{\epsilon}{{n}} + \frac{8 \mu^2r}{n}\left(7 \twonorm{\No} + \frac{8 n}{\sqrt{r}} \infnorm{\No}\right), \mbox{ and }\\
\supp{\widehat{\S}} &\subseteq \supp{\So}.
\end{align*}
\end{lemma}
\begin{proof}
Recall that in the $k^{\textrm{th}}$ stage, the update $\Ltn$ is given by: $\Ltn=P_k(\M-\St)$ and $\Stn$ is given by: $\Stn=H_\zeta(\M-\Ltn)$. 
Also, recall that $\Et \defas \So-\St$ and $\Etn \defas \So-\Stn$.

We prove the lemma by induction on both $k$ and $t$.
For the base case ($k=1$ and $t=-1$), we first note that the first inequality on $\infnorm{\Lt[0]-\Lo}$ is trivially satisfied.
Due to the thresholding step (step $3$ in Algorithm~\ref{algo:sap}) and the incoherence assumption on $\Lo$, we have: 
\begin{align*}
  \infnorm{\E^{(0)}} &\leq \frac{8\mu^2 r}{n}\left(\sigma_2^* + 2 \sigma_1^*\right), \mbox{ and }\\
  \supp{\Et[0]} &\subseteq \supp{\So}.
\end{align*}
So the base case of induction is satisfied.

We first do the inductive step over $t$ (for a fixed $k$).
By inductive hypothesis we assume that:
a) $\infnorm{\Et} \leq \frac{8\mu^2 r}{n}\left(|\sigma_{k+1}^*|+\left(\frac{1}{2}\right)^{t-1}|\sigma_k^*|
+ 7 \twonorm{\No} + \frac{8n}{\sqrt{r}} \infnorm{\No}\right)$,
b) $\supp{\Et}\subseteq \supp{\So}$. Then by Lemma~\ref{lem:dec1_noise}, we have:
\begin{align*}
\infnorm{\Ltn-\Lo} \leq \frac{2\mu^2 r}{n}\left(|\sigma_{k+1}^*|+\left(\frac{1}{2}\right)^{t}|\sigma_k^*|
+ 7 \twonorm{\No} + \frac{8n}{\sqrt{r}} \infnorm{\No}\right).
\end{align*}

Lemma~\ref{lem:supp-bound_noise} now tells us that
\begin{enumerate}
  \item	$\infnorm{\Etn} \leq  \frac{7\mu^2 r}{n}\left(|\sigma_{k+1}^*|+\left(\frac{1}{2}\right)^{t}|\sigma_k^*|
  + 7 \twonorm{\No} + \frac{8n}{\sqrt{r}} \infnorm{\No} \right)$, and
  \item	$\supp{\Etn} \subseteq \supp{\So}$.
\end{enumerate}
This finishes the induction over $t$. Note that we show a stronger bound than necessary on $\infnorm{\Etn}$.

We now do the induction over $k$. Suppose the hypothesis holds for stage $k$.
Let $T$ denote the number of iterations in each stage. We first obtain a lower bound on $T$.
Since
\begin{align*}
\twonorm{\M-\S^{(0)}} \geq \twonorm{\Lo +\No} - \twonorm{\E^{(0)}} \geq \abs{\sigma_1^*} - \alpha n \infnorm{\E^{(0)}}
\geq \frac{3}{4} \abs{\sigma_1^*},
\end{align*}
we see that $T \geq 10 \log \left(3 \mu^2 r \abs{\sigma_1^*}/\epsilon\right)$.
So, at the end of stage $k$, we have:
\begin{enumerate}
  \item	$\infnorm{\E^{(T)}} \leq \frac{7\mu^2 r}{n}\left(|\sigma_{k+1}^*|+\left(\frac{1}{2}\right)^{T}|\sigma_k^*|
  + 7 \twonorm{\No} + \frac{8n}{\sqrt{r}} \infnorm{\No} \right)
	\leq \frac{7\mu^2 r \abs{\sigma_{k+1}^*}}{n} + \frac{\epsilon}{10 n}$, and
  \item	$\supp{\Et[T]} \subseteq \supp{\So}$.
\end{enumerate}
Lemmas~\ref{lem:spec1} and~\ref{lem:weyl-perturbation} tell us that
$\abs{\sigma_{k+1}\left(\M-\St[T]\right) - \abs{\sigma_{k+1}^*}} \leq \twonorm{\E^{(T)}}
\leq \alpha\left( 7\mu^2 r \abs{\sigma_{k+1}^*} + \epsilon \right)$.
We will now consider two cases:
\begin{enumerate}
  \item	\textbf{Algorithm~\ref{algo:alg1} terminates:} This means that $\beta \sigma_{k+1}\left(\M-\St[T]\right) <
\frac{\epsilon}{2 n}$ which then implies that $\abs{\sigma_{k+1}^*} < \frac{\epsilon}{6 \mu^2 r}$. So we have:
\begin{align*}
\infnorm{\widehat{L} - \Lo} &= \infnorm{\Lt[T] - \Lo}%
	\leq \frac{2 \mu^2 r}{n}  \left(|\sigma_{k+1}^*|+\left(\frac{1}{2}\right)^{T}|\sigma_k^*|
	+ 7 \twonorm{\No} + \frac{8n}{\sqrt{r}} \infnorm{\No} \right) \\
	&\leq \frac{\epsilon}{5n} + \frac{2\mu^2 r}{n}\left(7 \twonorm{\No} + \frac{8n}{\sqrt{r}} \infnorm{\No}\right).
\end{align*}
This proves the statement about $\widehat{L}$. A similar argument proves the claim on $\infnorm{\widehat{S}-\So}$.
The claim on $\supp{\widehat{S}}$ follows since $\supp{\Et[T]} \subseteq \supp{\So}$.
  \item	\textbf{Algorithm~\ref{algo:alg1} continues to stage $\bm{\left(k+1\right)}$:} This means that
$\beta \sigma_{k+1}\left(\Lt[T]\right) \geq \frac{\epsilon}{2 n}$ which then implies that
$\abs{\sigma_{k+1}^*} > \frac{\epsilon}{8 \mu^2 r}$. So we have:
\begin{align*}
\infnorm{\E^{(T)}} &\leq \frac{7\mu^2 r}{n}\left(|\sigma_{k+1}^*|+\left(\frac{1}{2}\right)^{T}|\sigma_k^*|
+ 7 \twonorm{\No} + \frac{8n}{\sqrt{r}} \infnorm{\No}\right) \\
	&\leq \frac{7\mu^2 r}{n}\left(|\sigma_{k+1}^*|+ \frac{\epsilon}{10\mu^2 r n}
	+ 7 \twonorm{\No} + \frac{8n}{\sqrt{r}} \infnorm{\No} \right) \\
	&\leq \frac{7\mu^2 r}{n}\left(|\sigma_{k+1}^*|+ \frac{8 \abs{\sigma_{k+1}^*}}{10 n}
	+ 7 \twonorm{\No} + \frac{8n}{\sqrt{r}} \infnorm{\No} \right) \\
	&\leq \frac{8\mu^2 r}{n}\left(|\sigma_{k+2}^*|+ 2 \abs{\sigma_{k+1}^*}
	+ 7 \twonorm{\No} + \frac{8n}{\sqrt{r}} \infnorm{\No} \right).
\end{align*}
Similarly for $\infnorm{\Lt[T]-\Lo}$.
\end{enumerate}
This finishes the proof.
\end{proof}

\begin{proof}[Proof of Theorem~\ref{thm:noise}]
Using Lemma~\ref{lem:dec_noisy}, it suffices to show that the general case can be reduced to the case of symmetric
matrices. We will now outline this reduction.

Recall that we are given an $m \times n$ matrix $\M=\Lo + \No + \So $ where $\Lo$ is the true low-rank matrix, $\No$ dense corruption matrix
and $\So$ the sparse error matrix. Wlog, let $m \leq n$ and suppose $ \beta m \leq n < (\beta+1) m$, for some $\beta \geq 1$. We then consider the symmetric matrices
\begin{align}
&  \widetilde{\M} =
\underbrace{
\left[ 	\begin{array}{c}
		\begin{array}{c}
			0 \\ \vdots \\ 0
		\end{array}
		\cdots 
		\begin{array}{c}
			0 \\ \vdots \\ 0
		\end{array}	\\
		\trans{M} \cdots \trans{M}
	\end{array}
\right.}_{\beta \mbox{ times}}
\left.	\begin{array}{c}
		\begin{array}{c}
			M \\ \vdots \\ M
		\end{array} \\
		0
      	\end{array}
\right],
  \widetilde{\L} =
\underbrace{
\left[ 	\begin{array}{c}
		\begin{array}{c}
			0 \\ \vdots \\ 0
		\end{array}
		\cdots 
		\begin{array}{c}
			0 \\ \vdots \\ 0
		\end{array}	\\
		\trans{\left(\Lo\right)} \cdots \trans{\left(\Lo\right)}
	\end{array}
\right.}_{\beta \mbox{ times}}
\left.	\begin{array}{c}
		\begin{array}{c}
			\Lo \\ \vdots \\ \Lo
		\end{array} \\
		0
      	\end{array}
\right],\nonumber\\
&\qquad\qquad\qquad\quad  \widetilde{\N} =
\underbrace{
\left[ 	\begin{array}{c}
		\begin{array}{c}
			0 \\ \vdots \\ 0
		\end{array}
		\cdots 
		\begin{array}{c}
			0 \\ \vdots \\ 0
		\end{array}	\\
		\trans{\left(\No\right)} \cdots \trans{\left(\No\right)}
	\end{array}
\right.}_{\beta \mbox{ times}}
\left.	\begin{array}{c}
		\begin{array}{c}
			\Lo \\ \vdots \\ \Lo
		\end{array} \\
		0
      	\end{array}
\right],
\label{eqn:symmetrization_app_noise}
\end{align}
and $\wS = \wM - \wL$. A simple calculation shows that $\wL$ is incoherent with parameter $\sqrt{3}\mu$,
$\widetilde{\N}$ satisfies the assumption of Theorem~\ref{thm:noise} and $\wS$ satisfies
the sparsity condition (S1) with parameter $\frac{\alpha}{\sqrt{2}}$.
Moreover the iterates of \ncralgo~ with input $\wM$ have similar expressions as in \eqref{eqn:symmetrization_app_noise} in terms of the corresponding
iterates with input $\M$. This means that it suffices to obtain the same guarantees for Algorithm~\ref{algo:alg1}
for the symmetric case. Lemma~\ref{lem:dec_noisy} does precisely this, proving the theorem.
\end{proof}

\begin{figure}[t]
\centering
\begin{tabular}{ccc}
\hspace*{-10pt}\includegraphics[width=0.33\textwidth]{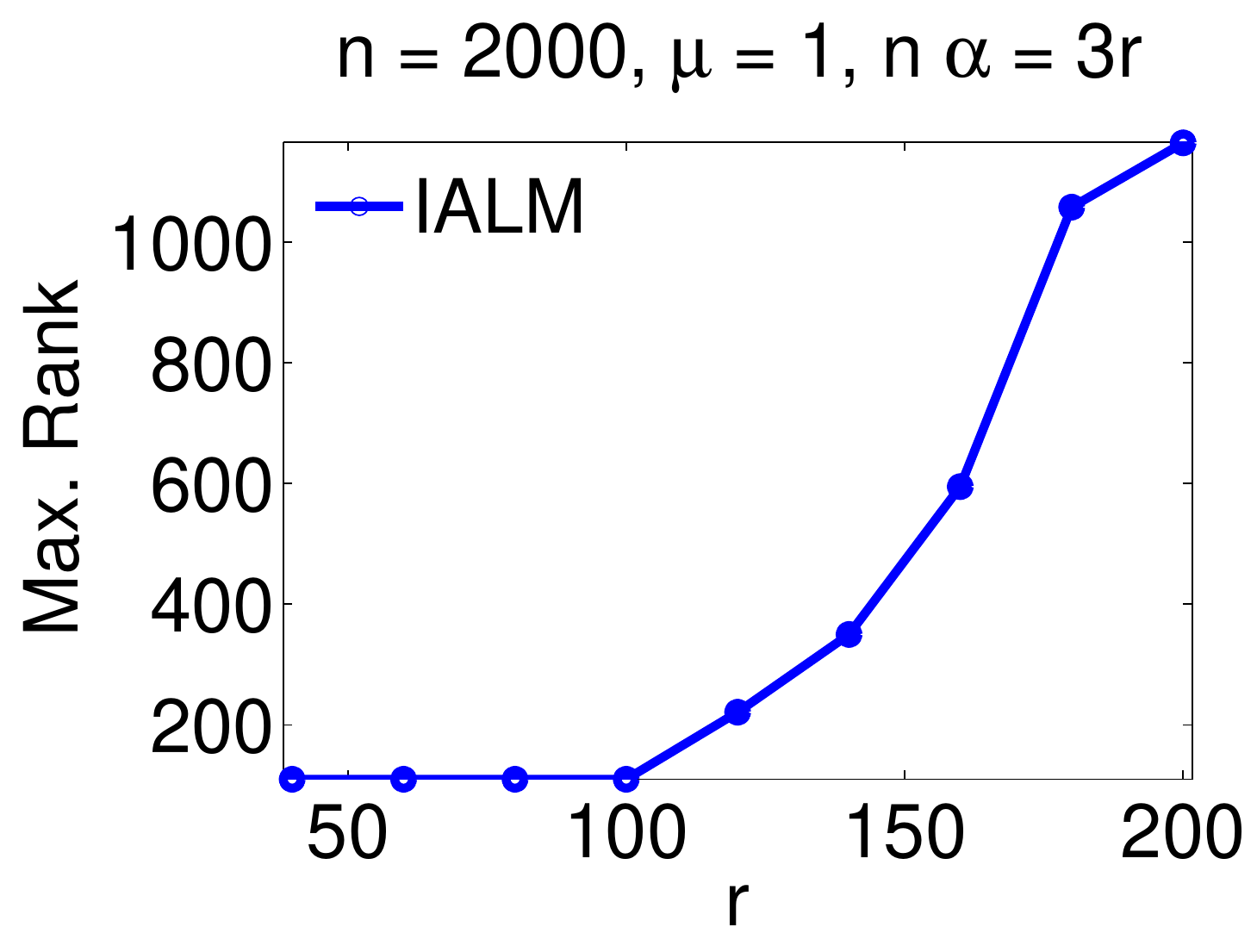}&
\hspace*{-10pt}\includegraphics[width=0.33\textwidth]{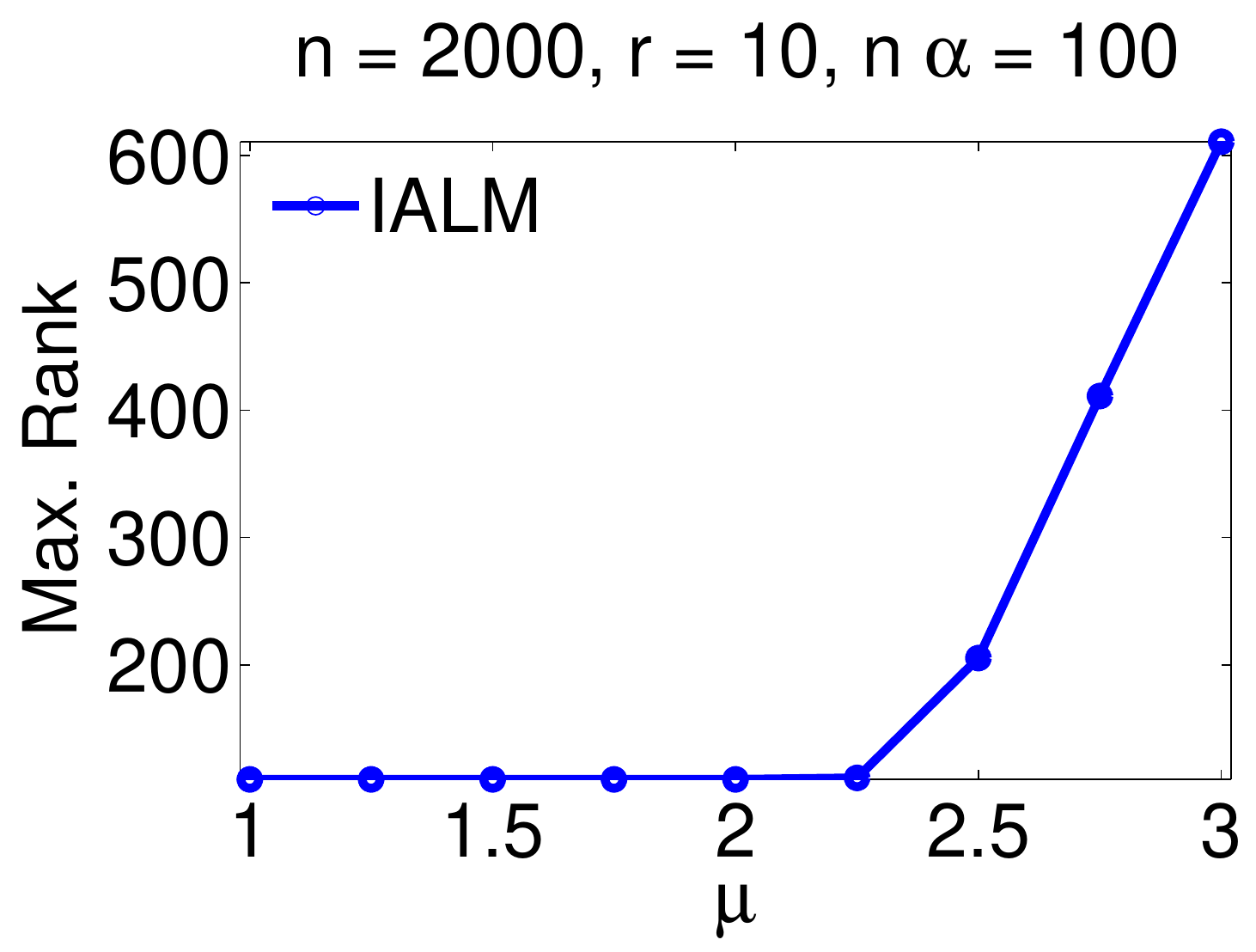}&
\hspace*{-10pt}\includegraphics[width=0.33\textwidth]{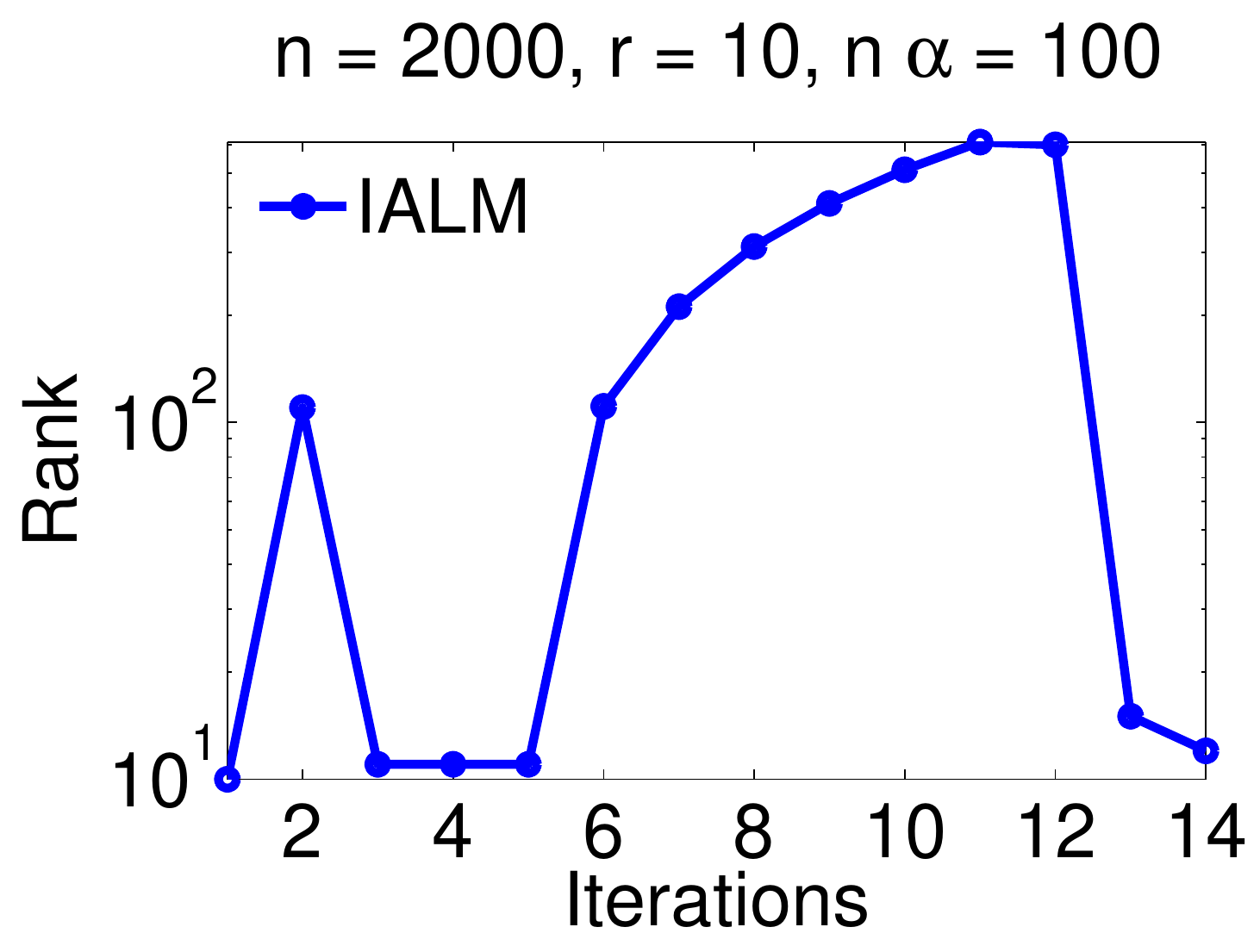}\\[-2pt]
(a)&(b)&(c)
\end{tabular}
\caption{ (a): Variation of the maximum rank of the intermediate low-rank solutions of IALM with rank. (b): Variation of the maximum rank of the intermediate low-rank solutions of IALM with incoherence. (c): Rank of the intermediate iterates of IALM for a particular run with $n=2000, r=10, \alpha=100/n, \mu=3$. Note that while the rank of the final output is $10$, intermediate iterates have rank as high as $800$.}
\label{fig:synthetic_plots_2}
\end{figure}

\begin{figure}[h!]
\centering
\begin{tabular}{ccc}
\includegraphics[width=0.3\textwidth]{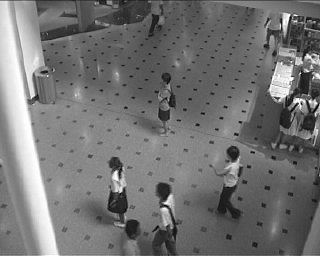}&
\includegraphics[width=0.3\textwidth]{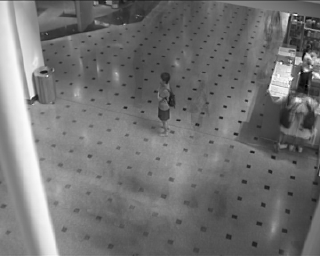}&
\includegraphics[width=0.3\textwidth]{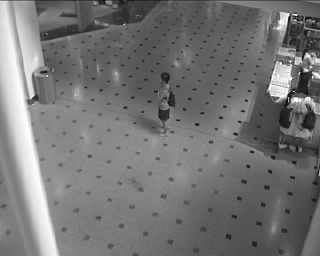}\\
(a)&(b)&(c)\\
\includegraphics[width=0.3\textwidth]{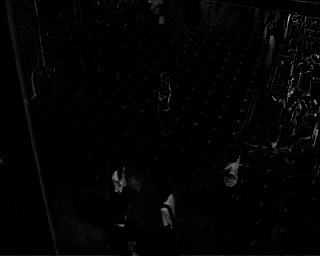}&
\includegraphics[width=0.3\textwidth]{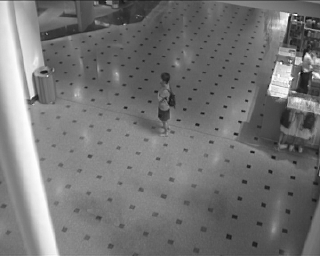}&
\includegraphics[width=0.3\textwidth]{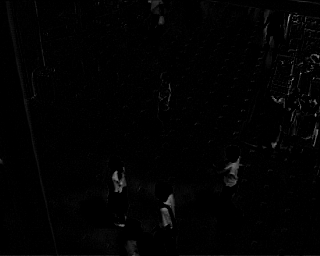}\\
(d)&(e)&(f)\\
\end{tabular}
\caption{Foreground-background separation in the {\em Shopping Mall} video. (a): Original frame in the video given as a part of the input to NcRPCA and IALM. (b): Corresponding frame from the best rank-$20$ approximation obtained using vanilla PCA; time taken for computing the low-rank approximation is $8.8s$. (c): Corresponding frame from the low-rank part obtained using NcRPCA; time taken by NcRPCA to compute the low-rank and sparse solutions is $292.1s$. (d): Corresponding frame from the sparse part obtained using NcRPCA. (e): Corresponding frame from the low-rank part obtained using IALM; time taken by IALM to compute the low-rank and sparse solutions is $783.4s$. (f): Corresponding frame from the sparse part obtained using IALM.}
\label{fig:shopping_mall}
\end{figure}

\begin{figure}[h!]
\centering
\begin{tabular}{ccc}
\includegraphics[width=0.3\textwidth]{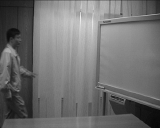}&
\includegraphics[width=0.3\textwidth]{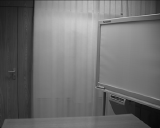}&
\includegraphics[width=0.3\textwidth]{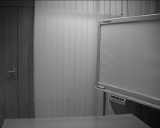}\\
(a)&(b)&(c)\\
\includegraphics[width=0.3\textwidth]{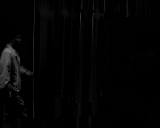}&
\includegraphics[width=0.3\textwidth]{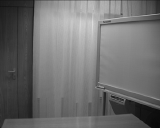}&
\includegraphics[width=0.3\textwidth]{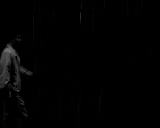}\\
(d)&(e)&(f)\\
\end{tabular}
\caption{Foreground-background separation in the {\em Curtain} video. (a): Original image frame in the video given as a part of the input to NcRPCA and IALM. (b): Corresponding frame from the best rank-$10$ approximation obtained using vanilla PCA; time taken for computing the low-rank approximation is $2.8s$. (c): Corresponding frame from the low-rank part obtained using NcRPCA; time taken by NcRPCA to compute the low-rank and sparse solutions is $39.5s$. (d): Corresponding frame from the sparse part obtained using NcRPCA. (e): Corresponding frame from the low-rank part obtained using IALM; time taken by IALM to compute the low-rank and sparse solutions is $989.0s$. (f): Corresponding frame from the sparse part obtained using IALM.}
\label{fig:curtain}
\end{figure}

\section{Additional experimental results}
\label{app:exps}
\paragraph{Synthetic datasets: }Extending Figure~\ref{fig:synthetic_plots_1}, the plots in Figure~\ref{fig:synthetic_plots_2} illustrate the point that soft thresholding, i.e., the convex relaxation approach, leads to intermediate solutions with high ranks. 
Figures~\ref{fig:synthetic_plots_2} (a)-(b) show the variation of the maximum rank of the intermediate low-rank solutions of IALM with rank and incoherence respectively; the results are averaged over 5 runs of the algorithm; we note that as the problem becomes harder, the maximum intermediate rank via soft thresholding (convex approach) increases, and this leads to higher running times. As an example of this phenomenon, Figure~\ref{fig:synthetic_plots_2} (c) shows the rank of the intermediate iterates of IALM for a particular run with $n = 2000, r =10, \alpha = 100/n, \mu = 3$; here, while the rank of the final output is $10$, intermediate iterates have a rank as high as $800$.
We run our synthetic simulations on a machine with Intel Dual 8-core Xeon (E5-2650) 2.0GHz CPU with 192GB RAM. 

\paragraph{Real-world datasets: }We provide some additional results concerning foreground-background separation in videos~\footnote{The datasets are available at \url{http://perception.i2r.a-star.edu.sg/bk_model/bk_index.html}}. We compare NcRPCA with IALM, and also with the low-rank solution obtained using vanilla PCA; we report the solutions obtained by NcRPCA and IALM methods for decomposing $\M$ into $\L+\S$ up to a relative error ($\|\M-\L-\S\|_F/\|\M\|_F$) of $10^{-3}$. We report the rank and the sparsity of the solutions obtained by the two methods along with the computational time. As mentioned before, the observed matrix $\M$ is formed by vectorizing each frame and stacking them column-wise. For illustration purposes, we arbitrarily select one of the original frames in the sequence of image frames obtained from the video, i.e., one of the columns of $\M$, and the corresponding columns in $\L$ and $\S$ obtained using NcRPCA and IALM. We run our real data experiments on a machine with Intel Dual 8-core Xeon (E5-2650) 2.0GHz CPU with 192GB RAM.\\ 

\emph{Shopping Mall dataset: }Figure~\ref{fig:shopping_mall} shows the comparison of NcRPCA and IALM on the ``Shopping Mall'' dataset which has $1286$ frames at a resolution of $256 \times 320$. NcRPCA achieves a solution of better visual quality (for example, unlike NcRPCA, notice the artifact of the low-rank solution from IALM in the top right corner of the image where the person is walking over the reflection of a light source; also notice the shadows of people in the low-rank part obtained by IALM which are not present in the low-rank solution obtained by NcRPCA), in $292.1s$, 
compared to IALM, which takes $783.4s$ 
until convergence. NcRPCA obtains a rank $20$ solution for $L$ with $\| S \|_0 = 95411896$ whereas IALM obtains a rank $286$ solution for $L$ with $\| S \|_0 = 86253965$. 

\emph{Curtain dataset: }We illustrate our recovery on one of the frames (frame $2773$) wherein a person enters a room with a curtain on the background. Figure~\ref{fig:curtain} shows the comparison of NcRPCA and IALM on the ``Curtain'' dataset which has $2964$ frames at a resolution of $160 \times 128$. NcRPCA achieves a solution, in $39.5s$, 
which is of similar visual quality to that of IALM, which takes $989.0s$ 
until convergence. NcRPCA obtains a rank $1$ solution for $L$ with $\| S \|_0 = 53897769$ whereas IALM obtains a rank $701$ solution for $L$ with $\| S \|_0 = 42310582$.


\begin{thebibliography}{WHML13}

\bibitem[AAJ{\etalchar{+}}13]{AgarwalEtal:SparseCoding2013}
A.~Agarwal, A.~Anandkumar, P.~Jain, P.~Netrapalli, and R.~Tandon.
\newblock {Learning Sparsely Used Overcomplete Dictionaries via Alternating
  Minimization}.
\newblock {\em Available on arXiv:1310.7991}, Oct. 2013.

\bibitem[AGH{\etalchar{+}}12]{AnandkumarEtal:tensor12}
A.~Anandkumar, R.~Ge, D.~Hsu, S.~M. Kakade, and M.~Telgarsky.
\newblock {Tensor Methods for Learning Latent Variable Models}.
\newblock {\em Available at arXiv:1210.7559}, Oct. 2012.

\bibitem[ANW12]{agarwal2012noisy}
A.~Agarwal, S.~Negahban, and M.~Wainwright.
\newblock Noisy matrix decomposition via convex relaxation: Optimal rates in
  high dimensions.
\newblock {\em The Annals of Statistics}, 40(2):1171--1197, 2012.

\bibitem[Bha97]{bhatia}
Rajendra Bhatia.
\newblock {\em Matrix Analysis}.
\newblock Springer, 1997.

\bibitem[{Che}13]{2013arXiv1310.0154C}
Y.~{Chen}.
\newblock {Incoherence-Optimal Matrix Completion}.
\newblock {\em ArXiv e-prints}, October 2013.

\bibitem[CLMW11]{CandesLMW11}
Emmanuel~J. Cand{\`e}s, Xiaodong Li, Yi~Ma, and John Wright.
\newblock Robust principal component analysis?
\newblock {\em J. ACM}, 58(3):11, 2011.

\bibitem[CSPW11]{ChandrasekaranSPW11}
Venkat Chandrasekaran, Sujay Sanghavi, Pablo~A. Parrilo, and Alan~S. Willsky.
\newblock Rank-sparsity incoherence for matrix decomposition.
\newblock {\em SIAM Journal on Optimization}, 21(2):572--596, 2011.

\bibitem[CSX12]{chen2012clustering}
Yudong Chen, Sujay Sanghavi, and Huan Xu.
\newblock Clustering sparse graphs.
\newblock In {\em Advances in neural information processing systems}, pages
  2204--2212, 2012.

\bibitem[EKYY13]{ErdosKYY13}
L{\'a}szl{\'o} Erd{\H{o}}s, Antti Knowles, Horng-Tzer Yau, and Jun Yin.
\newblock Spectral statistics of {E}rd{\H{o}}s--{R}{\'e}nyi graphs {I}: Local
  semicircle law.
\newblock {\em The Annals of Probability}, 2013.

\bibitem[Har13]{hardt2013provable}
Moritz Hardt.
\newblock On the provable convergence of alternating minimization for matrix
  completion.
\newblock {\em arXiv:1312.0925}, 2013.

\bibitem[HKZ11]{hsu2011robust}
Daniel Hsu, Sham~M Kakade, and Tong Zhang.
\newblock Robust matrix decomposition with sparse corruptions.
\newblock {\em ITIT}, 2011.

\bibitem[JNS13]{jain2013low}
Prateek Jain, Praneeth Netrapalli, and Sujay Sanghavi.
\newblock Low-rank matrix completion using alternating minimization.
\newblock In {\em STOC}, 2013.

\bibitem[KC12]{kyrillidis2012matrix}
Anastasios Kyrillidis and Volkan Cevher.
\newblock Matrix alps: Accelerated low rank and sparse matrix reconstruction.
\newblock In {\em SSP Workshop}, 2012.

\bibitem[Kes12]{Keshavan2012}
Raghunandan~H. Keshavan.
\newblock Efficient algorithms for collaborative filtering.
\newblock Phd Thesis, Stanford University, 2012.

\bibitem[LCM10]{lin2010augmented}
Zhouchen Lin, Minming Chen, and Yi~Ma.
\newblock The augmented lagrange multiplier method for exact recovery of
  corrupted low-rank matrices.
\newblock {\em arXiv:1009.5055}, 2010.

\bibitem[LHGT04]{li2004statistical}
Liyuan Li, Weimin Huang, IY-H Gu, and Qi~Tian.
\newblock Statistical modeling of complex backgrounds for foreground object
  detection.
\newblock {\em ITIP}, 2004.

\bibitem[MZYM11]{MobahiZYM11}
Hossein Mobahi, Zihan Zhou, Allen~Y. Yang, and Yi~Ma.
\newblock Holistic 3d reconstruction of urban structures from low-rank
  textures.
\newblock In {\em ICCV Workshops}, pages 593--600, 2011.

\bibitem[NJS13]{Netrapalli0S13}
Praneeth Netrapalli, Prateek Jain, and Sujay Sanghavi.
\newblock Phase retrieval using alternating minimization.
\newblock In {\em NIPS}, pages 2796--2804, 2013.

\bibitem[SAJ14]{SedeghiEtal:ADMM14}
H.~Sedghi, A.~Anandkumar, and E.~Jonckheere.
\newblock {Guarantees for Stochastic ADMM in High Dimensions}.
\newblock {\em Preprint.}, Feb. 2014.

\bibitem[Shi13]{shi2013sparse}
Lei Shi.
\newblock Sparse additive text models with low rank background.
\newblock In {\em Advances in Neural Information Processing Systems}, pages
  172--180, 2013.

\bibitem[WHML13]{wang2013solving}
X.~Wang, M.~Hong, S.~Ma, and Z.~Luo.
\newblock Solving multiple-block separable convex minimization problems using
  two-block alternating direction method of multipliers.
\newblock {\em arXiv:1308.5294}, 2013.

\bibitem[XCS12]{XuCS12}
Huan Xu, Constantine Caramanis, and Sujay Sanghavi.
\newblock Robust pca via outlier pursuit.
\newblock {\em IEEE Transactions on Information Theory}, 58(5):3047--3064,
  2012.

\end{thebibliography}
\end{document}